%% file: main.tex
  \documentclass[a4paper,10pt]{article}

\usepackage{floatflt}

\input{macros}

\usepackage{graphicx}

\title{Infinitary Lambda Calculi from a Linear Perspective\\(Long Version)}
\date{}
\author{Ugo Dal Lago\footnote{Universit\`a di Bologna \& INRIA Sophia Antipolis}}

\begin{document}
\maketitle
\begin{abstract}
We introduce a linear infinitary $\lambda$-calculus, called $\LLwot$,
in which two exponential modalities are available, the first one being
the usual, finitary one, the other being the only construct
interpreted coinductively. The obtained calculus embeds the infinitary
applicative $\lambda$-calculus and is universal for computations over
infinite strings. What is particularly interesting about $\LLwot$, is
that the refinement induced by linear logic allows to restrict both
modalities so as to get calculi which are \emph{terminating} inductively and
\emph{productive} coinductively. We exemplify this idea by analysing a
fragment of $\ell\Lambda$ built around the principles of \SLL\ and
\FLL. Interestingly, it enjoys confluence, contrarily to what happens
in ordinary infinitary $\lambda$-calculi.
\end{abstract}
\section{Introduction}
The $\lambda$-calculus is a widely accepted model of higher-order
functional programs---it faithfully captures functions and their
evaluation.  Through the many extensions introduced in the last forty
years, the $\lambda$-calculus has also been shown to be a model of
imperative features, control~\cite{Parigot92LPAR}, etc.  Usually, this
requires extending the class of terms with new operators, then
adapting type systems in such a way that ``good'' properties (e.g.,
confluence, termination, etc.), and possibly new ones, hold.

This also happened when potentially infinite structures came into
play. Streams and, more generally, coinductive data
have found their place in the $\lambda$-calculus, following the advent
of lazy data structures in functional programming languages like
\Haskell\ and \ML. By adopting lazy data structures, the programmer
has a way to represent infinity by finite means, relying on the fact
that data are not completely evaluated, but accessed ``on-demand''. In
presence of infinite structures, the usual termination property takes
the form of \emph{productivity}~\cite{Dijkstra80}: even if evaluating
a stream expression globally takes infinite time, looking for the next
symbol in it takes finite time.

All this has indeed been modelled in a variety of ways by enriching
$\lambda$-calculi and related type theories. One can cite, among the
many different works in this area, the ones by
Parigot~\cite{Parigot92TCS}, Raffalli~\cite{Raffalli93CSL}, Hughes at
al.~\cite{Hughes96POPL}, or Abel~\cite{Abel07APLAS}. Terms are
\emph{finite} objects, and the infinite nature of streams is modelled
through a staged inspection of so-called thunks. The key ingredient in
obtaining productivity and termination is usually represented by
sophisticated type systems.

There is also another way of modelling infinite structures in
$\lambda$-calculi, namely \emph{infinitary rewriting}, where both
terms and reduction sequences are not necessarily finite, and as a
consequence infinity is somehow internalised into the calculus.
Infinitary rewriting has been studied in the context of various
concrete rewrite systems, including first-order term
rewriting~\cite{Kennaway91RTA}, but also systems of higher-order
rewriting~\cite{Ketema11IC}.  In the case of
$\lambda$-calculus~\cite{Kennaway97TCS}, the obtained model is very
powerful, but does not satisfy many of the nice properties its finite
siblings enjoy, including confluence and finite developments, let
alone termination and productivity (which are anyway unexpected in the
absence of types).

In this paper, we take a fresh look at infinitary rewriting, through
the lenses of linear logic~\cite{Girard87TCS}. More specifically, we
define an untyped, linear, infinitary $\lambda$-calculus called
$\LLwot$, and study its basic properties, how it relates to
$\Lwot$~\cite{Kennaway97TCS}, and its expressive power. As expected,
incepting linearity does not allow \emph{by itself} to solve any of
the issues described above: $\LLwot$ embeds $\Lwotp{0}{0}{1}$, a
subsystem of $\Lwot$, and as such does not enjoy confluence. On the
other hand, linearity provides the right tools to \emph{control}
infinity: we delineate a simple fragment of $\LLwot$ which has all the
good properties one can expect, including productivity and confluence.
Remarkably, this is not achieved through types, but rather by purely
structural constraints on the way copying is managed, which is
responsible for its bad behaviour. Expressivity, on the other hand,
is not sacrificed.
\subsection{Linearity vs. Infinity}
The crucial r\^ole linearity has in this paper can be understood
without any explicit reference to linear logic as a proof system, but
through a series of observations about the usual, finitary,
$\lambda$-calculus. In any $\lambda$-term $\termone$, the variable
$\varone$ can occur free more than once, or not at all. If the
variable $\varone$ occurs free exactly once in $\termone$
\emph{whenever} we form an abstraction $\la{\varone}{\termone}$, the
$\lambda$-calculus becomes strongly normalising: at any rewrite step
the size of the term strictly decreases.  On the other hand, the
obtained calculus has a poor expressive power and, \emph{in this
  form}, is not the model of any reasonable programming language. Let
duplication reappear, then, but in controlled form: besides a
\emph{linear} abstraction $\la{\varone}{\termone}$, (which is well
formed only if $\varone$ occurs free exactly once in $\termone$) there
is also a \emph{nonlinear} abstraction $\na{\varone}{\termone}$, which
poses no constraints on the number of times $\varone$ occurs in
$\termone$. Moreover, and this is the crucial step, interaction with a
nonlinear abstraction is restricted: the argument to a nonlinear
application must itself be marked as \emph{duplicable} (and
\emph{erasable}) for $\beta$ to fire. This is again implemented by
enriching the category of terms with a new construct: given a term
$\termone$, the \emph{box} $\nm{\termone}$ is a duplicable version of
$\termone$. Summing up, the language at hand, call it $\LLwotfin$, is
built around applications, linear abstractions, nonlinear
abstractions, and boxes. Moreover, $\beta$-reduction comes in two
flavours:
$$
\ap{(\la{\varone}{\termone})}{\termtwo}\rightarrow\sbst{\termone}{\varone}{\termtwo};\qquad\qquad
\ap{(\na{\varone}{\termone})}{\nm{\termtwo}}\rightarrow\sbst{\termone}{\varone}{\termtwo}.
$$ What did we gain by rendering duplicating and erasing operations
explicit? Not much, apparently, since pure $\lambda$-calculus can be
embedded into $\LLwotfin$ in a very simple way following the so-called
Girard's translation~\cite{Girard87TCS}: abstractions become nonlinear
abstractions, and any application $\ap{\termone}{\termtwo}$ is
translated into $\ap{\termone}{\nm{\termtwo}}$. In other words, all
arguments can be erased and copied, and we are back to the world of
wild, universal, computation.  Not everything is lost, however: in any
nonlinear abstraction $\na{\varone}{\termone}$, $\varone$ can occur
any number of times in $\termone$, but at different \emph{depths}:
there could be linear occurrences of $\varone$, which do not lie in
the scope of any box, i.e., at depth $0$, but also occurrences of
$\varone$ at greater depths. Restricting the depths at which the bound
variable can occur in nonlinear abstractions gives rise to strongly
normalising fragments of $\LLwotfin$. Some of them can be seen,
through the Curry-Howard correspondence, as subsystems of linear logic
characterising complexity classes. This includes light linear
logic~\cite{Girard98IC}, elementary linear logic~\cite{Danos03IC} or
soft linear logic~\cite{Lafont04TCS}. As an example, the exponential
discipline of elementary linear logic can be formulated as follows in
$\LLwotfin$: for every nonlinear abstraction $\na{\varone}{\termone}$
all occurrences of $\varone$ in $\termone$ are at depth
$1$. Noticeably, this gives rise to a characterisation of elementary
functions. Similarly, soft linear linear can be seen as the fragment
of $\LLwot$ in which all occurrences of $\varone$ (if more than one)
occur at depth $0$ in any nonlinear abstraction
$\na{\varone}{\termone}$.

But why could all this be useful when rewriting becomes infinitary?
Infinitary $\lambda$-terms~\cite{Kennaway97TCS} are nothing more than
infinite terms defined according to the grammar of the
$\lambda$-calculus. In other words, one can have a term $\termone$
such that $\termone$ is syntactically equal to
$\la{\varone}{\ap{\varone}{\termone}}$. Evaluation can still be
modelled through $\beta$-reduction.  Now, however, reduction sequences
of infinite length (sometime) make sense: take
$\ap{\fpc}{(\la{\varone}{\la{\vartwo}{\ap{\vartwo}{\varone}}})}$,
where $\fpc$ is a fixed point combinator. It rewrites in infinitely
many steps to the infinite term $\termtwo$ such that
$\termtwo=\la{\vartwo}{\ap{\vartwo}{\termtwo}}$. Are all infinite
reduction sequences acceptable? Not really: an infinite reduction
sequence is said to \emph{converge} to the term $\termone$ only if,
informally, any finite approximation to $\termone$ can be reached
along a finite prefix of the sequence. In other words, reduction
should be applied deeper and deeper, i.e., the \emph{depth} of the
fired redex should tend to infinity. But how could one define the
depth of a subterm's occurrence in a given term? There are many
alternatives here, since, informally, one can choose to let the depth
increase (or not) when entering the body of an abstraction or any of
the two arguments of an application. Indeed, \emph{eight} different
calculi can be formed, each with different properties. For example,
$\Lwotp{0}{0}{1}$ is the calculus in which the depth only increases
while entering the argument position in applications, while in
$\Lwotp{1}{0}{0}$ the same happens when crossing abstractions. The
choice of where the depth increases is crucial not only when defining
infinite reduction sequences, but also when defining \emph{terms},
whose category is obtained by completing the set of finite terms with
respect to a certain metric. So, not all infinite terms are
well-formed. In $\Lwotp{0}{0}{1}$, as an example, the term
$\termone=\ap{\varone}{\termone}$ is well-formed, while
$\termone=\la{\varone}{\termone}$ is not. In $\Lwotp{1}{0}{0}$, the
opposite holds. In all the obtained calculi, however, many of the
properties one expects are not true: the Complete Developments Theorem
(i.e. the infinitary analogue of the Finite Complete
Theorem~\cite{Barendregt}) does not hold and, moreover, confluence
fails except if formulated in terms of so-called B\"ohm
reduction~\cite{Kennaway97TCS}. The reason is that $\Lwot$ is even
wilder than $\Lambda$: various forms of infinite computations can
happen, but only some of them are benign.  Is it that linearity could
help in taming all this complexity, similarly to what happens in the
finitary case? This paper gives a first positive answer to this
question.

\subsection{Contributions}
The system we will study in the rest of this paper, called $\LLwot$,
is obtained by incepting ideas from infinitary calculi into
$\LLwotfin$. Not one but \emph{two} kinds of boxes are available in
$\LLwot$, and the depth increases only while crossing boxes of the
second kind (called \emph{coinductive} boxes), while boxes of the
first kind (dubbed \emph{inductive}) leave the depth unchanged. As a
consequence, boxes are as usual the only terms which can be duplicated
and erased, but they are also responsible for the infinitary nature of
the calculus: any term not containing (coinductive) boxes is
necessarily finite. Somehow, the depths in the sense of $\Lwot$ and of
$\LLwotfin$ coincide.

Besides introducing $\LLwot$ and proving its basic properties, this
paper explores the expressive power of the obtained computational
model, showing that it suffers from the same problems affecting
$\Lwot$, but also that it has precisely the same expressive power as
that of Type-2 Turing machines, the reference computational model of
so-called computable analysis~\cite{Weihrauch00}.

The most interesting result among those we give in this paper consists
in showing that, indeed, a simple fragment of $\LLwot$, called
$\LLLwot$, is on the one hand flexible enough to encode streams and
guarded recursion on them, and on the other guarantees
productivity. Remarkably, confluence holds, contrary
to what happens for $\LLwot$. Actually, $\LLLwot$ is defined around
the same principles which lead to the definition of light logics. Each
kind of box, however, follows a distinct discipline: inductive boxes
are handled as in Lafont's \SLL~\cite{Lafont04TCS}, while coinductive
boxes follow the principles of \FLL~\cite{Danos03IC}, hence the name
of the calculus. So far, the \FLL's exponential
discipline has not been shown to have any computational meaning: now
we know that beyond it there is a form of guarded corecursion.

\section{$\LLwot$ and its Basic Properties}\label{sect:basicproperties}
In this section, we introduce a linear infinitary $\lambda$-calculus
called $\LLwot$, which is the main object of study of this paper. Some
of the dynamical properties of $\LLwot$ will be investigated. Before
defining $\LLwot$, some preliminaries about formal systems with
\emph{both} inductive and coinductive rules will be given.

\subsection{Mixing Induction and Coinduction in Formal Systems}
A formal system over a set of judgments $\jsone$ is given
  by a finite set of rules, all of them having one conclusions and an
  arbitrary finite number of premises. The rules of a \emph{mixed}
formal system $\sysone$ are of two kinds: those which are to be
interpreted inductively and those which are to be interpreted
\emph{co}inductively.  To distinguish between the two, inductive rules will
be denoted as usual, with a single line, while coinductive rules will
be indicated as follows: 
$
\infer= 
{\sjudg{\exptwo}} 
{\sjudg{\expone_1}& \ldots & \sjudg{\expone_n}}
$.  
Intuitively, any correct derivation in such a system is an
\emph{infinite} tree built following inductive and coinductive rules
where, however, any \emph{infinite} branch crosses coinductive rule
instances \emph{infinitely} often.  In other words, there cannot be
any infinite branch where, from a certain point on, only inductive
rules occur.

Formally, the set $\cd{\sysone}$ of derivable assertions of a mixed
formal system $\sysone$ over the set of judgments $\jsone$ can be seen
as the greatest fixpoint of the function
$\indrul{\sysone}\circ\coindrul{\sysone}$ where
$\coindrul{\sysone}:\powset{\jsone}\rightarrow\powset{\jsone}$ is the
monotone function induced by the the application of a \emph{single}
coinductive rule, while $\indrul{\sysone}$ is the function induced by
the application of inductive rules an arbitrary \emph{but finite}
number of times.  $\indrul{\sysone}$ can itself be obtained as the
least fixpoint of a monotone functional on the space of monotone
functions on $\powset{\jsone}$. The existence of $\cd{\sysone}$ can be
formally justified by Knaster-Tarski theorem, since all the involved
spaces are complete lattices, and the involved functionals are
monotone. This is, by the way, very close to the approach from~\cite{EndrullisPolonsky}. 
Let $\powset{\jsone}\leadsto\powset{\jsone}$ be the set of
monotone functions on the powerset $\powset{\jsone}$. A relation
$\sqsubseteq$ on $\powset{\jsone}\leadsto\powset{\jsone}$ can be
easily defined by stipulating that $\funcone\sqsubseteq\functwo$ iff
for every $\setone\subseteq\jsone$ it holds that
$\funcone(\setone)\subseteq\functwo(\setone)$.  The structure
$(\powset{\jsone}\leadsto\powset{\jsone},\sqsubseteq)$ is actually a
complete lattice, because:
\begin{varitemize}
\item
  The relation $\sqsubseteq$ is a partial order. In particular, antisymmetry is a consequence of
  function extensionality: if for every $\setone$, both $\funcone(\setone)\subseteq\functwo(\setone)$
  and $\functwo(\setone)\subseteq\funcone(\setone)$, then $\funcone$ and $\functwo$ are the same
  function.
\item
  Given a set of monotone functions $\funsetone$, its lub and sup exist and are the functions
  $\funcone,\functwo:\powset{\jsone}\leadsto\powset{\jsone}$ such that for every
  $\setone\subseteq\jsone$,
  $$
    \funcone(\setone)=\bigcap_{\functhree\in\funsetone}\functhree(\setone);\qquad\qquad
    \functwo(\setone)=\bigcup_{\functhree\in\funsetone}\functhree(\setone).
  $$
  It is easy to verify that $\funcone$ is monotone, that it minorises
  $\funsetone$, and that it majorises any minoriser of $\funsetone$. Similarly
  for $\functwo$.
\end{varitemize}
The function $\indrul{\sysone}$ is the least fixpoint of the monotone functional
$\functione$ on $\powset{\jsone}\leadsto\powset{\jsone}$ defined as follows: to
every function $\funcone$, $\functione$ associates the function $\functwo$ obtained
by feeding $\funcone$ with the argument set $\setone$, and then applying one additional
instance of inductive rules from $\sysone$. Since $(\powset{\jsone}\leadsto\powset{\jsone},\sqsubseteq)$ 
is a complete lattice, $\indrul{\sysone}$ is guaranteed to exist.

Formal systems in which \emph{all} rules are either coinductively or
inductively interpreted have been studied extensively (see,
e.g. \cite{Leroy09IC}). Our constructions, although relatively
simple, do not seem to have appeared before at least in this form. The conceptually
closest work is the one by Endrullis and coauthors~\cite{EndrullisPolonsky}.

How could we prove anything about $\cd{\sysone}$? How should we proceed, as an example, when
tyring to prove that a given subset $\setone$ of $\jsone$ is included in $\cd{\sysone}$? Fixed-point
theory tells us that the correct way to proceed consists in showing that $\setone$ is
$(\indrul{\sysone}\circ\coindrul{\sysone})$-consistent, namely that $\setone\subseteq
\indrul{\sysone}(\coindrul{\sysone}(\setone))$. We will frequently apply this proof strategy in
the following.
\subsection{An Infinitary Linear Lambda Calculus}
\newcommand{\pof}{\triangleright}
\emph{Preterms} are potentially infinite terms built from the
following grammar:
$$
\termone,\termtwo::=\;\varone\midd\ap{\termone}{\termone}\midd\la{\varone}{\termone}\midd\ia{\varone}{\termone}\midd
   \ca{\varone}{\termone}\midd\im{\termone}\midd\cm{\termone},
$$
where $\varone$ ranges over a denumerable set $\varsetone$ of
variables. $\ptms$ is the set of preterms.  
The notion of capture-avoiding substitution of a preterm $\termone$ for a
variable $\varone$ in anoter preterm $\termtwo$, denoted
$\sbst{\termtwo}{\varone}{\termone}$, can be defined, this time by
\emph{coinduction}, on the structure of $\termtwo$:
\begin{align*}
\sbst{(\varone)}{\varone}{\termone}&=\termone;\\
\sbst{(\vartwo)}{\varone}{\termone}&=\vartwo;\\
\sbst{(\la{\varone}{\termthree})}{\vartwo}{\termone}&=\la{\varone}{\sbst{\termthree}{\vartwo}{\termone}};
   \qquad\qquad\mbox{if $\varone\not\in\FV{\termone}$}\\
\sbst{(\ia{\varone}{\termthree})}{\vartwo}{\termone}&=\ia{\varone}{\sbst{\termthree}{\vartwo}{\termone}};
   \qquad\qquad\mbox{if $\varone\not\in\FV{\termone}$}\\
\sbst{(\ca{\varone}{\termthree})}{\vartwo}{\termone}&=\ca{\varone}{\sbst{\termthree}{\vartwo}{\termone}};
   \qquad\qquad\mbox{if $\varone\not\in\FV{\termone}$}\\
\sbst{(\ap{\termthree}{\termfour})}{\vartwo}{\termone}&=
   \ap{(\sbst{\termthree}{\vartwo}{\termone})}{(\sbst{\termfour}{\vartwo}{\termone})};\\
\sbst{(\im{\termthree})}{\vartwo}{\termone}&=\im{(\sbst{\termthree}{\vartwo}{\termone})};\\
\sbst{(\cm{\termthree})}{\vartwo}{\termone}&=\cm{(\sbst{\termthree}{\vartwo}{\termone})}.
\end{align*}
Observe that all the equations above are guarded, so this is a well-posed definition. 
An \emph{inductive} (respectively, \emph{coinductive}) \emph{box}
is any preterm in the form $\im{\termone}$ (respectively, $\cm{\termone}$).

Please notice that any (guarded) equation has a unique solution over
preterms. As an example, $\termone=\la{\varone}{\termone}$,
$\termtwo=\ap{\termtwo}{(\la{\varone}{\varone})}$, and
$\termone=\ap{\vartwo}{\cm{\termone}}$ all have unique solutions.
In other words, infinity is everywhere. Only certain preterms,
however, will be the objects of this study.  To define the class of
``good'' preterms, simply called \emph{terms}, we now introduce a
mixed formal system. An \emph{environment} $\conone$ is simply a set
of expressions (called \emph{patterns}) in one the following three
forms:
$$
\patone::=\varone\midd\im{\varone}\midd\cm{\varone},
$$
where any variable occurs in at most \emph{one} pattern in $\conone$. If $\conone$ and $\contwo$ are two
disjoint environments, then $\conone,\contwo$ is their union. An environment
is \emph{linear} if it only contains variables. Linear environments are indicated with metavariables
like $\lconone$ or $\lcontwo$.
A \emph{term judgment} is an expression in the form $\tjudg{\conone}{\termone}$ where $\conone$ is an environment and
$\termone$ is a term. A \emph{term} is any preterm $\termone$
for which a judgment $\tjudg{\conone}{\termone}$ can be derived by the formal system $\LLwot$, whose rules are
in Figure~\ref{fig:llwotwfr}.
\begin{figure*}
\begin{center}
\fbox{
\begin{minipage}{.97\textwidth}
\begin{center}
{\footnotesize
$$
\infer[\LLvarl]
{\tjudg{\im{\lconone},\cm{\lcontwo},\varone}{\varone}}
{}
\qquad
\infer[\LLvari]
{\tjudg{\im{\lconone},\cm{\lcontwo},\im{\varone}}{\varone}}
{}
\qquad
\infer[\LLvarc]
{\tjudg{\im{\lconone},\cm{\lcontwo},\cm{\varone}}{\varone}}
{}
\qquad
\infer[\LLapp]
{\tjudg{\conone,\contwo,\im{\lconone},\cm{\lcontwo}}{\ap{\termone}{\termtwo}}}
{
  \tjudg{\conone,\im{\lconone},\cm{\lcontwo}}{\termone}
  &
  \tjudg{\contwo,\im{\lconone},\cm{\lcontwo}}{\termtwo}
}
$$
$$
\infer[\LLlaml]
{\tjudg{\conone}{\la{\varone}{\termone}}}
{\tjudg{\conone,\varone}{\termone}}
\qquad
\infer[\LLlami]
{\tjudg{\conone}{\ia{\varone}{\termone}}}
{\tjudg{\conone,\im{\varone}}{\termone}}
\qquad
\infer[\LLlamc]
{\tjudg{\conone}{\ca{\varone}{\termone}}}
{\tjudg{\conone,\cm{\varone}}{\termone}}
\qquad
\infer[\LLlmi]
{\tjudg{\im{\lconone},\cm{\lcontwo}}{\im{\termone}}}
{\tjudg{\im{\lconone},\cm{\lcontwo}}{\termone}}
\qquad
\infer=[\LLlmc]
{\tjudg{\im{\lconone},\cm{\lcontwo}}{\cm{\termone}}}
{\tjudg{\im{\lconone},\cm{\lcontwo}}{\termone}}
$$}
\end{center}
\vspace{2pt}
\end{minipage}}
\end{center}
\caption{$\LLwot$: Well-Formation Rules.}\label{fig:llwotwfr}
\end{figure*}
Please notice that $\LLlmc$ is coinductive, while all the other rules
are inductive. This means that on terms, differently from preterms,
not all recursive equations have a solution, anymore.  As an example
$\termone=\ap{\vartwo}{\cm{\termone}}$ is a term: a derivation
$\derone$ for $\tjudg{\im{\vartwo}}{\termone}$ can be found in
Figure~\ref{fig:ed}.
\begin{figure*}
\begin{center}
\fbox{
\begin{minipage}{.97\textwidth}
\begin{center}
$$
\infer[\LLapp]
{\tjudg{\im{\vartwo}}{\termone}}
{
  \infer[\LLvarc]
    {\tjudg{\im{\vartwo}}{\vartwo}}
    {}
  &
  \infer=[\LLlmc]
    {\tjudg{\im{\vartwo}}{\cm{\termone}}}
    {\derone\;\pof\tjudg{\im{\vartwo}}{\termone}}
}
\qquad\qquad
\infer[\LLapp]
{\tjudg{\emcon}{\ap{\termtwo}{(\la{\varone}{\varone})}}}
{
  \dertwo\pof\tjudg{\emcon}{\ap{\termtwo}{(\la{\varone}{\varone})}}
  &
  \infer[\LLlaml]
    {\tjudg{\emcon}{\la{\varone}{\varone}}}
    {
      \infer[\LLvarl]
      {\tjudg{\varone}{\varone}}
      {}
    }
}
$$
\end{center}
\vspace{0pt}
\end{minipage}}
\end{center}
\caption{$\LLwot$: Example Derivation Trees $\derone$ (left) and $\dertwo$ (right).}\label{fig:ed}
\end{figure*}
$\derone$ is indeed a well-defined derivation because although
infinite, any infinite path in it contains infinitely many occurrences
of $\LLlmc$.  If we try to proceed in the same way with the preterm
$\termtwo=\ap{\termtwo}{(\la{\varone}{\varone})}$, we immediately
fail: the only candidate derivation $\dertwo$ looks like the one in
Figure~\ref{fig:ed} and is not well-defined (it contains a ``loop'' of
inductive rule instances). We write $\tms{\LLwot}(\conone)$ for the
set of all preterms $\termone$ such that $\tjudg{\conone}{\termone}$.
The union of $\tms{\LLwot}(\conone)$ over all environments $\conone$
is denoted simply as $\tms{\LLwot}$.

Some observations about the r\^ole of environments are now in order: If
$\tjudg{\varone,\conone}{\termone}$, then $\varone$ necessarily
occurs free in $\termone$, but exactly once and in \emph{linear
position} (i.e. not in the scope of any box). If, on the other hand,
$\tjudg{\im{\varone},\conone}{\termone}$, then $\varone$ can occur
free any number of times, even infinitely often, in $\termone$.
Similarly when $\tjudg{\cm{\varone},\conone}{\termone}$. Observe, in
this respect, that inductive and coinductive boxes are actually very
permissive: if $\tjudg{\im{\varone},\conone}{\termone}$, $\varone$ can
even occur in the scope of coinductive boxes, while $\varone$ can
occur in the scope of inductive boxes if
$\tjudg{\cm{\varone},\conone}{\termone}$. We claim that this is
source of the great expressive power of the calculus, but also
of its main defects (e.g. the absence of confluence).

\newcommand{\bsone}{s} \newcommand{\bstwo}{t}
\newcommand{\embs}{\varepsilon}
\newcommand{\strnat}[1]{#1^{\bullet}}
\newcommand{\NFs}[1]{\mathsf{NFs}(#1)} 
Sometimes it is useful to denote symbols $\isym$ and $\csym$ in a
unified way. To that purpose, let $\BB$ be the set $\{0,1\}$ of binary
digits, which is ranged over by metavariables like $\bitone$ or
$\bittwo$.  $\psym{0}$ stands for $\isym$, while $\psym{1}$ is
$\csym$. For every $\bsone\in\BB^*$, we can define
\emph{$\bsone$-contexts}, ranged over by metavariables like
$\pctxone{\bsone}$ and $\pctxtwo{\bsone}$, as follows, by
induction on $\bsone$:
\begin{align*}
\pctxone{\varepsilon}::=\emctx;\qquad\qquad\pctxone{0\cdot\bsone}::=\im{\pctxone{\bsone}};\qquad\qquad
\pctxone{1\cdot\bsone}::=\cm{\pctxone{\bsone}};\\
\pctxone{\bsone}::=\ap{\pctxone{\bsone}}{\termone}\midd\ap{\termone}{\pctxone{\bsone}}
  \midd\la{\varone}{\pctxone{\bsone}}\midd\ia{\varone}{\pctxone{\bsone}}\midd\ca{\varone}{\pctxone{\bsone}}.
\end{align*}
Given any subset $\setone$ of $\BB^*$, an $\setone$-context, sometime
denoted as $\pctxone{\setone}$ is an $\bsone$-context where
$\bsone\in\setone$. A \emph{context} $\ctxone$ is simply any
$\bsone$-context $\pctxone{\bsone}$. For every natural
number $\natone\in\NN$, $\strnat{\natone}$ is the
set of those strings in $\BB^*$ in which $1$ occurs
precisely $\natone$ times. For every $\natone$, the language
$\strnat{\natone}$ is regular.
\subsection{Finitary and Infinitary Dynamics}\label{sect:fininfdyn}
In this section, notions of finitary and infinitary reduction for
$\LLwot$ are given.  \emph{Basic reduction} is a binary relation
$\redLLbas\subseteq\ptms\times\ptms$ defined by the following three
rules (where, as usual, $\termone\redLLbas\termtwo$ stands for
$(\termone,\termtwo)\in\;\redLLbas$):
$$
\ap{(\la{\varone}{\termone})}{\termtwo}\redLLbas\sbst{\termone}{\varone}{\termtwo};\quad
\ap{(\ia{\varone}{\termone})}{\im{\termtwo}}\redLLbas\sbst{\termone}{\varone}{\termtwo};\quad
\ap{(\ca{\varone}{\termone})}{\cm{\termtwo}}\redLLbas\sbst{\termone}{\varone}{\termtwo}.
$$
Basic reduction can be applied in any $\bsone$-context, giving
rise to a \emph{ternary} relation
$\redLLwod\subseteq\ptms\times\BB^*\times\ptms$, simply called
\emph{reduction}.  That is defined by stipulating that
$(\termone,\bsone,\termtwo)\in\redLLwod$ iff there are a
$\bsone$-context $\pctxone{\bsone}$ and two terms $\termthree$ and
$\termfour$ such that $\termthree\redLLbas\termfour$,
$\termone=\actx{\pctxone{\bsone}}{\termthree}$, and
$\termtwo=\actx{\pctxone{\bsone}}{\termfour}$.  In this case, the
reduction step is said to occur \emph{at level} $\bsone$ and we write
$\termone\redLL{\bsone}\termtwo$ and
$\level{\termone\redLL{\bsone}\termtwo}=\bsone$. We often employ the
notation $\redLL{\setone}$, i.e., $\redLL{\setone}$ is the union of
the relations $\redLL{\bsone}$ for all $\bsone\in\setone$.  If
$\termone\redLL{\bsone}\termtwo$ but we are not interested in the
specific $\bsone$, we simply write $\termone\redLLwod\termtwo$. If
$\termone\redLL{\strnat{\natone}}\termtwo$, then reduction is said to
occur at depth $\natone$.

Given $\setone\subseteq\BB^*$, a \emph{$\setone$-normal form}
is any term  $\termone$ such that whenever 
$(\termone,\bsone,\termtwo)$, it holds that $\bsone\not\in\setone$.
The set of all $\setone$-normal forms is denoted as
$\NFs{\setone}$. In the just introduced notations, the singleton
$\bsone$ is often used in place of $\{\bsone\}$ if this does
not cause any ambiguity. A \emph{normal form} is simply
a $\BB^*$-normal form.

Depths and levels have a different nature: while the depth increases
only when entering a coinductive box, the level changes while entering
any kind of box, and this is the reason why levels are binary strings
rather than natural numbers.

Since $\sbst{\termone}{\varone}{\termtwo}$ is well-defined whenever
$\termone$ and $\termtwo$ are preterms, reduction is certainly closed
as a relation on the space of preterms. That it is also closed on
terms is not trivial. First of all, substitution lemmas need to be
proved for the three kinds of patterns which possibly appear in
environments. The first of these lemmas concerns linear variables:
\begin{lemma}[Substitution Lemma, Linear Case]\label{lem:substlinllwot}
  If $\tjudg{\conone,\varone,\im{\lconone},\cm{\lcontwo}}{\termone}$
  and $\tjudg{\contwo,\im{\lconone},\cm{\lcontwo}}{\termtwo}$, then it holds that 
  $\tjudg{\conone,\contwo,\im{\lconone},\cm{\lcontwo}}{\sbst{\termone}{\varone}{\termtwo}}$.
\end{lemma}
\begin{proof}
We can prove that the following subset $\setone$ of judgments is consistent with $\LLwot$:
$$
\left\{
\tjudg{\conone,\contwo,\im{\lconone},\cm{\lcontwo}}{\sbst{\termone}{\varone}{\termtwo}}\midd
\tjudg{\conone,\varone,\im{\lconone},\cm{\lcontwo}}{\termone}\in\cd{\LLwot}\wedge
\tjudg{\contwo,\im{\lconone},\cm{\lcontwo}}{\termtwo}\in\cd{\LLwot}
\right\}\cup\cd{\LLwot}.
$$
Suppose that $\judgone=\tjudg{\conone,\contwo,\im{\lconone},\cm{\lcontwo}}{\sbst{\termone}{\varone}{\termtwo}}$
is in $\setone$. If $\judgone\in\cd{\LLwot}$, then
of course 
$$
\judgone\in\indrul{\LLwot}(\coindrul{\LLwot}(\cd{\LLwot}))\subseteq
\indrul{\LLwot}(\coindrul{\LLwot}(\setone)).
$$
Otherwise, we know that $\judgtwo=\tjudg{\conone,\varone,\im{\lconone},\cm{\lcontwo}}{\termone}\in\cd{\LLwot}$
and, by the fact $\cd{\LLwot}=\indrul{\LLwot}(\coindrul{\LLwot}(\cd{\LLwot}))$ we can infer
that $\judgtwo\in\indrul{\LLwot}(\coindrul{\LLwot}(\cd{\LLwot}))$, namely that $\judgtwo$
can be obtained by judgments in $\cd{\LLwot}$ by applying coinductive rules once, followed
by $\natone$ inductive rules. We prove that 
$\judgone\in\cd{\LLwot}$ by induction on $\natone$:
\begin{varitemize}
\item
  If $n=0$, then $\judgtwo$ can be obtained by means of the rule $\LLlmc$, but this
  is impossible since the environment in any judgment obtained this way cannot contain
  any variable, and $\judgtwo$ actually contains one.
\item
  If $n>0$, then we distinguish a number of cases, depending on the last inductive
  rule applied to derive $\judgtwo$:
  \begin{varitemize}
  \item
    If it is $\LLvarl$, then $\termone=\varone$, $\sbst{\termone}{\varone}{\termtwo}$
    is simply $\termtwo$ and $\conone$ does not contain any variable. By the fact that
    $$
    \tjudg{\contwo,\im{\lconone},\cm{\lcontwo}}{\termtwo}\in\cd{\LLwot}
    $$
    it follows, by a Weakening Lemma, that $\judgone\in\cd{\LLwot}$.
  \item
    It cannot be either $\LLvari$ or $\LLvarc$ or $\LLlmi$: in all these
    cases the underlying environment cannot contain variables;
  \item
    If it is either $\LLlaml$ or $\LLapp$, then the induction hypothesis yields the thesis 
    immediately;
  \item
    If it is either $\LLlami$ or $\LLlamc$, then a Weakening Lemma 
    applied to the induction hypothesis leads to the thesis.  
  \end{varitemize}
\end{varitemize}
From $\judgone\in\cd{\LLwot}$, it follows that
$$
\judgone\in\indrul{\LLwot}(\coindrul{\LLwot}(\cd{\LLwot}))
\subseteq\indrul{\LLwot}(\coindrul{\LLwot}(\setone)),
$$
which is the thesis. This concludes the proof.
\end{proof}
A similar result can be given when the substituted variable occurs in
the scope of an inductive box:
\begin{lemma}[Substitution Lemma, Inductive Case]\label{lem:substindllwot}
  If $\tjudg{\conone,\im{\varone},\im{\lconone},\cm{\lcontwo}}{\termone}$
  and $\tjudg{\im{\lconone},\cm{\lcontwo}}{\termtwo}$, then it holds that
  $\tjudg{\conone,\im{\lconone},\cm{\lcontwo}}{\sbst{\termone}{\varone}{\termtwo}}$.
\end{lemma}
\begin{proof}
The structure of this proof is identical to the one of Lemma~\ref{lem:substlinllwot}.
\end{proof}
When the variable is in the scope of a coinductive box, almost nothing changes:
\begin{lemma}[Substitution Lemma, Coinductive Case]\label{lem:substcoillwot}
  If $\tjudg{\conone,\im{\varone},\im{\lconone},\cm{\lcontwo}}{\termone}$
  and $\tjudg{\cm{\lconone},\cm{\lcontwo}}{\termtwo}$, then it holds that
  $\tjudg{\conone,\im{\lconone},\cm{\lcontwo}}{\sbst{\termone}{\varone}{\termtwo}}$.
\end{lemma}
\begin{proof}
The structure of this proof is identical to the one of Lemma~\ref{lem:substlinllwot}.
\end{proof}
The following is an analogue of the so-called Subject Reduction Theorem, and is
an easy consequence of substitution lemmas:
\begin{proposition}[Well-Formedness is Preseved by Reduction]
  If $\tjudg{\conone}{\termone}$ and $\termone\redLLwod\termtwo$, then $\tjudg{\conone}{\termtwo}$.
\end{proposition}
\begin{proof}
Let us first of all prove that if $\termone\redLLbas\termtwo$, and $\tjudg{\conone}{\termone}$, then
$\tjudg{\conone}{\termtwo}$. let us distinguish three cases:
\begin{varitemize}
\item
  If $\termone$ is $\ap{(\la{\varone}{\termthree})}{\termfour}$, then
  $\tjudg{\contwo,\varone,\im{\lconone},\cm{\lcontwo}}{\termthree}$ and
  $\tjudg{\conthree,\im{\lconone},\cm{\lcontwo}}{\termfour}$, where 
  $\conone=\contwo,\conthree,\im{\lconone},\cm{\lcontwo}$.
  By Lemma \ref{lem:substlinllwot}, one gets that
  $\tjudg{\conone}{\sbst{\termthree}{\varone}{\termfour}}$, which is the thesis.
\item
  If $\termone$ is $\ap{(\ia{\varone}{\termthree})}{\im{\termfour}}$, then
  $\tjudg{\contwo,\im{\varone},\im{\lconone},\cm{\lcontwo}}{\termthree}$ and
  $\tjudg{\im{\lconone},\cm{\lcontwo}}{\termfour}$, where 
  $\conone=\contwo,\im{\lconone},\cm{\lcontwo}$. By Lemma \ref{lem:substindllwot}, one gets that
  $\tjudg{\conone}{\sbst{\termthree}{\varone}{\termfour}}$, which is the thesis.
\item
  If $\termone$ is $\ap{(\ca{\varone}{\termthree})}{\cm{\termfour}}$, then
  $\tjudg{\contwo,\im{\lconone},\cm{\varone},\cm{\lcontwo}}{\termthree}$ and
  $\tjudg{\im{\lconone},\cm{\lcontwo}}{\termfour}$, where 
  $\conone=\contwo,\im{\lconone},\cm{\lcontwo}$. By Lemma \ref{lem:substcoillwot}, one gets that
  $\tjudg{\conone}{\sbst{\termthree}{\varone}{\termfour}}$, which is the thesis.
\end{varitemize}
One can then prove that for every context $\ctxone$ and for every
pair of terms $\termone$ and $\termtwo$ such that $\termone\redLLbas\termtwo$,
if $\tjudg{\conone}{\actx{\ctxone}{\termone}}$
then $\tjudg{\conone}{\actx{\ctxone}{\termtwo}}$. This is an induction on the
structure of $\ctxone$.
\end{proof}
Finitary reduction, as a consequence, is well-defined not only on
preterms, but also on terms.

How about \emph{infinite} reduction? Actually, even defining what an
infinite reduction sequence \emph{is} requires some care. In this paper,
following~\cite{EndrullisPolonsky}, we define infinitary reduction by
way of a mixed formal system (see
Section~\ref{sect:basicproperties}). The judgments of this formal
system have two forms, namely $\sjudg{\termone\redLLinf\termtwo}$ and
$\sjudg{\termone\redLLnext\termtwo}$, and its rules are in
Figure~\ref{fig:llwotid}.
\begin{figure*}
\begin{center}
\fbox{
\begin{minipage}{.97\textwidth}
\begin{center}
{\footnotesize
$$
\infer=[]
{\sjudg{\termone\redLLinf\termthree}}
{\termone\redLLwod^*\termtwo & \sjudg{\termtwo\redLLnext\termthree}}
\qquad
\infer[]
{\sjudg{\termone\termthree\redLLnext\termtwo\termfour}}
{\sjudg{\termone\redLLnext\termtwo} & \sjudg{\termthree\redLLnext\termfour}}
\qquad
\infer[]
{\sjudg{\varone\redLLnext\varone}}
{}
\qquad
\infer[]
{\sjudg{\la{\varone}{\termone}\redLLnext\la{\varone}{\termtwo}}}
{\sjudg{\termone\redLLnext\termtwo}}
$$
$$
\infer[]
{\sjudg{\ia{\varone}{\termone}\redLLnext\ia{\varone}{\termtwo}}}
{\sjudg{\termone\redLLnext\termtwo}}
\qquad
\infer[]
{\sjudg{\ca{\varone}{\termone}\redLLnext\ca{\varone}{\termtwo}}}
{\sjudg{\termone\redLLnext\termtwo}}
\qquad
\infer[]
{\sjudg{\im{\termone}\redLLnext\im{\termtwo}}}
{\sjudg{\termone\redLLnext\termtwo}}
\qquad
\infer[]
{\sjudg{\cm{\termone}\redLLnext\cm{\termtwo}}}
{\sjudg{\termone\redLLinf\termtwo}}
$$
}
\end{center}
\vspace{2pt}
\end{minipage}}
\end{center}
\caption{$\LLwot$: Infinitary Dynamics.}\label{fig:llwotid}
\end{figure*}
The relation $\redLLinf$ is the infinitary, coinductively defined,
notion of reduction we are looking for. Informally,
$\sjudg{\termone\redLLinf\termtwo}$ is provable (and we write,
simply, $\termone\redLLinf\termtwo$) iff there is a third term $\termthree$
such that $\termone$ reduces to $\termthree$ in a finite amount of
steps, and $\termthree$ itself reduces infinitarily to $\termtwo$
where, however, infinitary reduction is applied at depths higher
than one. The latter constraint is taken care of by $\redLLnext$.

An infinite reduction sequence, then, can be seen as being decomposed
into a finite prefix and finitely many infinite suffixes, each
involving subterm occurrences at higher depths. We claim that this
corresponds to strongly convergent reduction sequences as defined
in~\cite{Kennaway97TCS}, although a formal comparison is outside the
scope of this paper (see, however,~\cite{EndrullisPolonsky}).

What are the main properties of $\redLLinf$? Is it a confluent
notion of reduction?  Is it that $\termtwo$ is a normal form
whenever $\termone\redLLinf\termtwo$? Actually, the latter question can be
easily given a negative answer: take the unique preterm $\termone$
such that $\termone=\;\cm{(\ap{\termone}{(\ap{I}{I})})}$, where
$I=\la{\varone}{\varone}$ is the identity combinator. Of course,
$\tjudg{\emcon}{\termone}$. We can prove that both
$\termone\redLLinf\termtwo$ and that $\termone\redLLinf\termthree$,
where 
$$
\termtwo=\;\cm{(\ap{\cm{(\ap{\termtwo}{(\ap{I}{I})})}}{I})};\qquad
\termthree=\;\cm{(\ap{\cm{(\ap{\termthree}{I})}}{(\ap{I}{I})})}.
$$
(Infact, $\redLLinf$ is reflexive, see Lemma~\ref{lemma:reflextrans} below.)
Neither $\termtwo$ nor $\termthree$ is a normal form. It is easy to
realise that there is $\termfour$ to which both $\termtwo$ and
$\termthree$ reduces to, namely
$\termfour=\cm{(\ap{\cm{(\ap{\termfour}{I})}}{I})}$. Confluence, however, does
not hold in general, as can be easily shown by considering
the following two terms $\termone$ and $\termtwo$:
$$
\termone=K\cm{\termtwo}\cm{K};\qquad\termtwo=K\cm{\termone}\cm{I};
$$
where
$K=\ca{\varone}{\ca{\vartwo}{\varone}}$ and $I=\ca{\varone}{\varone}$.
If we reduce $\termone$ at even and at odd depths, we end up
at two terms $\termthree$ and $\termfour$ which cannot be joined
by $\redLLinf$, namely the following:
$$
\termthree=K\cm{\termthree}\cm{I};\qquad \termfour=K\cm{\termfour}\cm{K}.
$$
The deep reason why this phenomenon happens is an interference between
$\redLLwod$ and $\redLLnext$: there are $\termfive$ and $\termsix$
such that $\termone\redLLwod\termfive$ and
$\termone\redLLnext\termsix$, but there is no $\termseven$ such that
$\termfive\redLLnext\termseven$ and $\termsix\redLLwod\termseven$.

\subsection{Level-by-Level Reduction}\label{sect:levelbylevel}
One restriction of $\redLLinf$ that will be useful in the following
is the so called \emph{level-by-level} reduction, which is obtained
by constraining reduction to occur at deeper levels only if no redex
occurs at outer levels. Formally, let $\redLLwodlbl$ be the
restriction of $\redLLwod$ obtained by stipulating that
$(\termone,\bsone,\termtwo)\in\redLLwodlbl$ iff there are a
$\bsone$-context $\pctxone{\bsone}$ and two terms $\termthree$ and
$\termfour$ such that $\termthree\redLLbas\termfour$,
$\termone=\actx{\pctxone{\bsone}}{\termthree}$, and
$\termtwo=\actx{\pctxone{\bsone}}{\termfour}$, \emph{and moreover},
$\termone$ is $\bstwo$-normal for every prefix $\bstwo$ of $\bsone$.
Then one can obtain $\redLLnextlbl$ and $\redLLinflbl$ from
$\redLLwodlbl$ as we did for $\redLLnext$ and
$\redLLinf$ in Section~\ref{sect:fininfdyn} (see Figure~\ref{fig:llwotidlbl}).
\begin{figure*}
\begin{center}
\fbox{
\begin{minipage}{.97\textwidth}
\begin{center}
{\footnotesize
$$
\infer=[]
{\sjudg{\termone\redLLinflbl\termthree}}
{\termone\redLLwodlbl^*\termtwo & \sjudg{\termtwo\redLLnextlbl\termthree} & \termtwo\in\NFs{\strnat{0}}}
\qquad
\infer[]
{\sjudg{\termone\termthree\redLLnextlbl\termtwo\termfour}}
{\sjudg{\termone\redLLnextlbl\termtwo} & \sjudg{\termthree\redLLnextlbl\termfour}}
\qquad
\infer[]
{\sjudg{\varone\redLLnextlbl\varone}}
{}
\quad
\infer[]
{\sjudg{\la{\varone}{\termone}\redLLnextlbl\la{\varone}{\termtwo}}}
{\sjudg{\termone\redLLnextlbl\termtwo}}
$$
$$
\infer[]
{\sjudg{\ia{\varone}{\termone}\redLLnextlbl\ia{\varone}{\termtwo}}}
{\sjudg{\termone\redLLnextlbl\termtwo}}
\quad
\infer[]
{\sjudg{\ca{\varone}{\termone}\redLLnextlbl\ca{\varone}{\termtwo}}}
{\sjudg{\termone\redLLnextlbl\termtwo}}
\quad
\infer[]
{\sjudg{\im{\termone}\redLLnextlbl\im{\termtwo}}}
{\sjudg{\termone\redLLnextlbl\termtwo}}
\quad
\infer[]
{\sjudg{\cm{\termone}\redLLnextlbl\cm{\termtwo}}}
{\sjudg{\termone\redLLinflbl\termtwo}}
$$
}
\end{center}
\vspace{2pt}
\end{minipage}}
\end{center}
\caption{$\LLwot$: Level-by-Level Infinitary Dynamics.}\label{fig:llwotidlbl}
\end{figure*}

Clearly, if $\termone\redLLinflbl\termtwo$, then
$\termone\redLLinf\termtwo$. Moreover, $\redLLinflbl$, contrarily to
$\redLLinf$, is confluent, simply because $\redLLwodlbl$ satisfies a
diamond-property. This is not surprising, and has been already observed in
the realm of finitary rewriting~\cite{Simpson05}.  Moreover,
level-by-level is effective: only a finite portion of
$\termone$ needs to be inspected in order to check if a given redex
occuring in $\termone$ can be fired (or to find one if $\termone$
contains one). Indeed, it will used to define what it means for
a term in $\LLwot$ to compute a function, which is the main topic
of the following section.

\section{On the Expressive Power of $\LLwot$}
The just introduced calculus $\LLwot$ can be seen as a refinement of
$\Lwot$ obtained by giving a first-order status to depths, i.e., by
introducing a specific construct which makes the depth to increase
when crossing it. In this section, we will give an interesting result
about the absolute expressive power of the introduced calculus: not
only functions on finite strings can be expressed, but also functions
on \emph{infinite} strings.  Before doing that, we will investigate on
the possibility to embed existing infinitary $\lambda$-calculi from
the literature.
\subsection{Embedding $\Lwot$}
Some introductory words about $\Lwot$ are now in order
(see~\cite{Kennaway97TCS} or \cite{Kennaway03} for more
details). Originally $\Lwot$ has been defined based 
on completing the space of $\lambda$-terms with respect
to a metric. Here we reformulate the calculus 
differently, based on coinduction. 

In $\Lwot$, there are many choices as to where the underlying depth
can increase. Indeed, \emph{eight} different calculi can be
defined. More specifically, for every
$\bitone,\bittwo,\bitthree\in\BB$,
$\Lwotp{\bitone}{\bittwo}{\bitthree}$ is obtained by stipulating that:
\begin{varitemize}
\item
 the depth increases while crossing abstractions iff $\bitone=1$; 
\item
 the depth increases when going through the first argument of an application iff $\bittwo=1$;
\item
 the depth increases when entering the second argument of an application iff $\bitthree=1$.
\end{varitemize}
Formally, one can define terms of $\Lwotp{\bitone}{\bittwo}{\bitthree}$
as those (finite or infinite) $\lambda$-terms $\termone$ such that
$\ptjudg{\conone}{\bitone\bittwo\bitthree}{\termone}$ is derivable
through the rules in Figure~\ref{fig:lwotwfr}.
\begin{figure*}
\begin{center}
\fbox{
\begin{minipage}{.97\textwidth}
\begin{center}
{\footnotesize
$$
\infer[\Lvar]
{\ptjudg{\conone,\varone}{\bitone\bittwo\bitthree}{\varone}}
{}
\qquad
\infer[\Lapp]
{\ptjudg{\conone}{\bitone\bittwo\bitthree}{\ap{\termone}{\termtwo}}}
{
  \dptjudg{\conone}{\bitone\bittwo\bitthree}{\funsym}{\termone}
  &
  \dptjudg{\conone}{\bitone\bittwo\bitthree}{\argsym}{\termtwo}
}
\qquad
\infer[\Llam]
{\ptjudg{\conone}{\bitone\bittwo\bitthree}{\la{\varone}{\termone}}}
{
  \dptjudg{\varone,\conone}{\bitone\bittwo\bitthree}{\lamsym}{\termone}
}
$$
$$
\infer[\Largi]
{\dptjudg{\conone}{0\bittwo\bitthree}{\argsym}{\termone}}
{\ptjudg{\conone}{0\bittwo\bitthree}{\termone}}
\qquad
\infer=[\Largc]
{\dptjudg{\conone}{1\bittwo\bitthree}{\argsym}{\termone}}
{\ptjudg{\conone}{1\bittwo\bitthree}{\termone}}
\qquad
\infer[\Lfuni]
{\dptjudg{\conone}{\bitone 0 \bitthree}{\funsym}{\termone}}
{\ptjudg{\conone}{\bitone 0 \bitthree}{\termone}}
\qquad
\infer=[\Lfunc]
{\dptjudg{\conone}{\bitone 1 \bitthree}{\funsym}{\termone}}
{\ptjudg{\conone}{\bitone 1 \bitthree}{\termone}}
\qquad
\infer[\Llami]
{\dptjudg{\conone}{\bitone\bittwo 0}{\lamsym}{\termone}}
{\ptjudg{\conone}{\bitone\bittwo 0}{\termone}}
\qquad
\infer=[\Llamc]
{\dptjudg{\conone}{\bitone\bittwo 1}{\lamsym}{\termone}}
{\ptjudg{\conone}{\bitone\bittwo 1}{\termone}}
$$
}
\end{center}
\vspace{2pt}
\end{minipage}}
\end{center}
\caption{$\Lwot$: Well-Formation Rules.}\label{fig:lwotwfr}
\end{figure*}
Finite and infinite reduction sequences can be defined exactly as we
have just done for $\LLwot$. The obtained calculi have a
very rich and elegant mathematical theory. Not much is known, however,
about whether $\Lwot$ can be tailored as to guarantee key properties
of programs working on streams, like productivity.

Let us now show how $\Lwotp{0}{0}{0}$ and $\Lwotp{0}{0}{1}$ can indeed
be embedded into $\LLwot$. For every binary digit $\bitone$, the map
$\peremb{\cdot}{\bitone}$ from the space of terms of
$\Lwotp{0}{0}{\bitone}$ into the space of preterms is defined as
follows:
\begin{align*}
 \peremb{\varone}{\bitone}&=\varone;\\
 \peremb{\ap{\termone}{\termtwo}}{\bitone}&=\ap{\peremb{\termone}{\bitone}}{\pam{\bitone}{\peremb{\termtwo}{\bitone}}};\\
 \peremb{\ab{\varone}{\termone}}{\bitone}&=\pa{\bitone}{\varone}{\peremb{\termone}{\bitone}};
\end{align*}
where the expression $\pam{\bitone}{\termone}$ is defined to be $\im{\termone}$
if $\bitone=0$ and $\cm{\termone}$ if $\bitone=1$. Please observe that
$\peremb{\cdot}{\bitone}$ is defined by coinduction on the space of
terms of $\Lwotp{0}{0}{\bitone}$, which contains possibly infinite
objects. By the way, $\peremb{\cdot}{\bitone}$ can be seen as Girard's
embedding of intuitionistic logic into linear logic where, however,
the kind of boxes and abstractions we use depends on $\bitone$: we go
inductive if $\bitone=0$ and coinductive otherwise.

First of all, preterms obtained via the embedding are actuall
(depending on $\bitone$) in the environment:
\begin{lemma}
 For every $\termone\in\Lambda_{00\bitone}$, it holds that $\tjudg{\pam{\bitone}{\FV{\termone}}}{\peremb{\termone}{\bitone}}$.
\end{lemma}
\begin{proof}
  We proceed by showing that the following set of judgments $\setone$
  is consistent with $\LLwot$:
$$
\left\{\tjudg{\pam{\bitone}{\FV{\termone}}}{\peremb{\termone}{\bitone}}\midd\termone\in\Lwotp{0}{0}{\bitone}\right\}.
$$
Suppose that $\tjudg{\pam{\bitone}{\FV{\termone}}}{\peremb{\termone}{\bitone}}$, where
$\termone\in\Lwotp{0}{0}{\bitone}$. This implies that $\termone=\actx{\lctxone}{\termtwo_1,\ldots,\termtwo_\natone}$,
where $\termtwo_1,\ldots,\termtwo_\natone\in\Lwotp{0}{0}{\bitone}$. The fact that
$\termone\in\indrul{\LLwot}(\coindrul{\LLwot}(\setone))$ can be proved by induction on $\lctxone$, with
different cases depending on the value of $\bitone$.
\end{proof}
From a dynamical point of view, this embedding is \emph{perfect}: not only basic reduction in $\Lwot$ can be simulated
in $\LLwot$, but any reduction we do in $\peremb{\termone}{\bitone}$ can be traced back to a reduction happening
in $\termone$.
\begin{lemma}[Perfect Simulation]\label{lemma:perfsim}
 For every $\termone\in\Lwotp{0}{0}{\bitone}$, if $\termone\redlambda{\natone}\termtwo$, then
 $\peremb{\termone}{\bitone}\redLL{\strnat{\natone}}\peremb{\termtwo}{\bitone}$.
 Moreover, for every $\termone\in\Lambda_{00\bitone}$, if 
 $\peremb{\termone}{\bitone}\redLL{\strnat{\natone}}\termtwo$, then there
 is $\termthree$ such that $\termone\redlambda{\natone}\termthree$ and
 $\peremb{\termthree}{\bitone}\syneq\termtwo$.
\end{lemma}
\begin{proof}
Just consider how a redex $\ap{(\la{\varone}{\termone})}{\termtwo}$ in $\Lwotp{0}{0}{\bitone}$ is translated: it becomes
$\ap{(\pa{\bitone}{\varone}{\peremb{\termone}{\bitone}})}{\pam{\bitone}{\peremb{\termtwo}{\bitone}}}$.
As can be easily proved, for every $\bitone$ and for 
every $\termone,\termtwo\in\Lwotp{0}{0}{\bitone}$, 
$$
\sbst{(\peremb{\termone}{\bitone})}{\varone}{\peremb{\termtwo}{\bitone}}=
\peremb{\sbst{\termone}{\varone}{\termtwo}}{\bitone}.
$$
This means $\LLwot$ correctly simulates $\Lwotp{0}{0}{\bitone}$. For the converse, just observe that the only
redexes in $\peremb{\termone}{\bitone}$ are those corresponding to redexes from $\termone$.
\end{proof}
One may wonder whether $\Lwotp{0}{0}{1}$ is the only (non-degenerate) dialect of $\Lwot$ which
can be simulated in $\LLwot$. Actually, besides the perfect embedding we have just given there 
is also an \emph{imperfect} embedding of systems in the form $\Lwotp{\bitone}{0}{\bittwo}$
(where $\bitone$ and $\bittwo$ are binary digits) into $\LLwot$:
\begin{align*}
 \imperemb{\varone}{\bitone}{\bittwo}&=\varone;\\
 \imperemb{\ap{\termone}{\termtwo}}{\bitone}{\bittwo}&=\ap{(\pa{\bitone}{\varone}{\varone})}
 {(\ap{\imperemb{\termone}{\bitone}{\bittwo}}{\pam{\bittwo}{\imperemb{\termtwo}{\bitone}{\bittwo}}})};\\
 \imperemb{\ab{\varone}{\termone}}{\bitone}{\bittwo}&=\pa{\bittwo}{\varone}{\pam{\bittwo}{\imperemb{\termone}{\bitone}{\bittwo}}}.
\end{align*}
This is a variation on the so-called call-by-value embedding of
intuitionistic logic into linear logic (i.e. the embedding induced by
the map $(A\rightarrow B)^\bullet=!(A^\bullet)\multimap!(B^\bullet)$,
see~\cite{Maraist95ENTCS}). Please notice, however, that variables
occur nonlinearly in the environment, while the term itself is never
a box, contrarily to the usual call-by-value embedding (where at least
values are translated into boxes). As expected:
\begin{lemma}
  For every $\termone\in\Lambda_{\bitone0\bittwo}$, it holds that
  $\tjudg{\pam{\bittwo}{\FV{\termone}}}{\imperemb{\termone}{\bitone}{\bittwo}}$.
\end{lemma}
As can be easily realised, any $\beta$ step in $\Lwot$ can be
simulated by \emph{two} reduction steps in $\LLwot$. This makes the
simulation imperfect:
\begin{lemma}[Imperfect Simulation]
 For every $\termone\in\Lambda_{\bitone0\bittwo}$, if $\termone\redlambda{\natone}\termtwo$, then
 $\imperemb{\termone}{\bitone}{\bittwo}\redLL{\strnat{\natone}}^2\imperemb{\termtwo}{\bitone}{\bittwo}$.
\end{lemma}
\begin{lemma}
  Moreover, for every $\termone\in\Lambda_{\bitone0\bittwo}$, if
  $\imperemb{\termone}{\bitone}{\bittwo}\redLL{\strnat{\natone}}\termtwo$, then
  there is $\termthree$ such that
  $\termone\redlambda{\natone}\termthree$ and
  $\peremb{\termthree}{\bitone}\syneq\termtwo$.
\end{lemma}
\subsection{$\LLwot$ as a Stream Programming Language}\label{sect:llwotspl}
One of the challenges which lead to the introduction of infinitary
rewriting systems is, as we argued in the Introduction, the
possibility to inject infinity (as arising in lazy data
structures such as streams) into formalisms like the
$\lambda$-calculus. In this section, we show that, indeed, $\LLwot$
can not only express terms of any free (co)algebras, but also any
effective function on them. Moreover, anything $\LLwot$ can compute
can also be computed by Type-2 Turing machines. To the author's
knowledge, this is the first time this is done for any system of
infinitary rewriting (some partial results, however, can be found
in~\cite{Barendregt09}).

\subsection{Signatures and Free (Co)algebras}
A \emph{signature} $\sigone$ is a set of function symbols, each with
an associated \emph{arity}. Function symbols will be denoted with
metavariables like $\funsymone$ or $\funsymtwo$. In this paper, we are
concerned with \emph{finite} signatures, only. Sometimes, a signature
has a very simple structure: an \emph{alphabet signature} is a
signature whose function symbols all have arity $1$, except for a single
nullary symbol, denoted $\emstr$. Given an alphabet $\alpone$,
$\fins{\alpone}$ ($\infs{\alpone}$, respectively) denotes the set of
finite (infinite, respectively) words over
$\alpone$. $\fininfs{\alpone}$ is simply
$\fins{\alpone}\cup\infs{\alpone}$. For every alphabet $\alpone$,
there is a corresponding alphabet signature $\sigone_\alpone$.  Given
a signature (or an alphabet), one usually needs to define the set of
terms built according to the algebra itself. Indeed, the \emph{free
  algebra} $\fa{\sigone}$ induced by a signature $\sigone$ is the set
of all finite terms built from function symbols in $\sigone$, i.e.,
all terms \emph{inductively} defined from the following production
(where $\natone$ is the arity of $\funsymone$):
\begin{equation}\label{equ:algcoalg}
\ftermone::=\funsymone(\ftermone_1,\ldots,\ftermone_n).
\end{equation}
There is another canonical way of building terms from signatures,
however. One can interpret the production above \emph{coinductively},
getting a space of finite \emph{and infinite} terms: the \emph{free
  coalgebra} $\fc{\sigone}$ induced by a signature $\sigone$ is the
set of all finite \emph{and infinite} terms built from function
symbols in $\sigone$, following~(\ref{equ:algcoalg}).
Notice that $\fins{\alpone}$ is isomorphic to $\fa{\sigone_\alpone}$,
while $\fininfs{\alpone}$ is isomorphic to $\fc{\sigone_\alpone}$. We
often elide the underlying isomorphisms, confusing strings and
terms.

\subsection{Representing (Infinitary) Terms in $\LLwot$}.
There are many number systems which work well in the finitary
$\lambda$-calculus. One of them is the well-known system of Church
numerals, in which $\natone\in\NN$ is represented by
$\la{\varone}{\la{\vartwo}{\varone^\natone\vartwo}}$. We here adopt
another scheme, attributed to Scott~\cite{Wadsworth80}: this allows to
make the relation between depths and computation more explicit.  Let
$\sigone=\{\funsymone_1,\ldots,\funsymone_\natone\}$ and suppose that
symbols in $\sigone$ can be totally ordered in such a way that
$\funsymone_\natone\leq\funsymone_\nattwo$ iff $\natone\leq\nattwo$.
Terms of the free algebra $\fa{\sigone}$ can be encoded as terms of
$\LLwot$ as follows
$$
\faemb{\sigone}{\funsymone_\nattwo(\ftermone_1,\ldots,\ftermone_\natthree)}=
\ia{\varone_1}{\cdots.\ia{\varone_\natone}{\varone_\nattwo\im{\faemb{\sigone}{\ftermone_1}}\cdots\im{\faemb{\sigone}{\ftermone_\natthree}}}}.
$$ 
Similarly for terms in the free coalgebra $\fc{\sigone}$:
$$
\fcemb{\sigone}{\funsymone_\nattwo(\ftermone_1,\ldots,\ftermone_\natthree)}=
\ia{\varone_1}{\cdots.\ia{\varone_\natone}{\varone_\nattwo\cm{\fcemb{\sigone}{\ftermone_1}}\cdots\cm{\fcemb{\sigone}{\ftermone_\natthree}}}}.
$$  
Given a string $\strone\in\fins{\alpone}$, the term
$\faemb{\sigone_\alpone}{\strone}$ is denoted simply as
$\fsemb{\strone}$. Similarly, if $\strone\in\infs{\alpone}$,
$\isemb{\strone}$ indicates $\fcemb{\sigone_\alpone}{\strone}$. Please
observe how $\faemb{\sigone}{\cdot}$ differs from
$\fcemb{\sigone}{\cdot}$: in the first case the encoding of subterms
are wrapped in an inductive box, while in the second case the
enclosing box is coinductive. This very much reflects the spirit of
our calculus: in $\fcemb{\sigone}{\ftermone}$ the depth increases
whenever entering a subterm, while in $\faemb{\sigone}{\ftermtwo}$,
the depth \emph{never} increases.
\subsection{Universality}\label{sect:univers}
The question now is: given the encoding in the last paragraph, which
\emph{functions} can we represent in $\LLwot$? If domain and codomain
are \emph{free algebras}, a satisfactory answer easily comes from the
universality of ordinary $\lambda$-calculus with respect to
computability on finite structures: the class of functions at hand
coincides with the effectively computable ones. If, on the other hand,
functions handling or returning terms from free coalgebras are of
interest, the question is much more interesting.

The expressive power of $\LLwot$ actually coincides with the one of
Type-2 Turing machines: these are mild generalisations of ordinary
Turing machines obtained by allowing inputs and outputs to be
not-necessarily-finite strings. Such a machine consists of finitely
many, initially blank work tapes, finitely many one-way input tapes
and a single one-way output tape. Noticeably, input tapes initially
contain not-necessarily-finite strings, while the output tape is
sometime supposed to be filled with an infinite
string. See~\cite{Weihrauch00} for more details and
Figure~\ref{fig:twotm} for a graphical representation of the structure
of any Type-2 Turing machine: black arrows represent the data flow,
whereas grey arrows represent the possible direction of the head in
the various tapes.
\begin{figure}
  \begin{center}
\fbox{
\begin{minipage}{.445\textwidth}
\vspace{4pt}
\begin{center}
\includegraphics[scale=0.6]{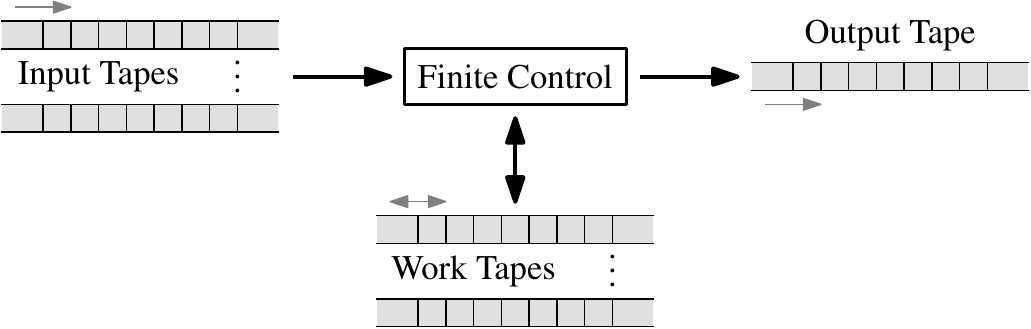}
\end{center}
\vspace{1pt}
\end{minipage}}
\end{center}
\caption{The Structure of a Type-2 Turing Machine}\label{fig:twotm}
\end{figure}

We now need to properly formalise \emph{when} a given function on possibly infinite strings
can be represented by a term in $\LLwot$. To that purpose, let $\SS$ be the set $\{*,\omega\}$, where the 
two elements of $\SS$ are considered merely as \emph{symbols}, with no internal structure. Objects in $\SS$ 
are indicated with metavariables like $\starone$ or $\startwo$. A partial function $\funone$ from 
$\pfininfs{\alpone}{\starone_1}\times\cdots\times\pfininfs{\alpone}{\starone_n}$ to 
$\pfininfs{\alpone}{\startwo}$ is said to be \emph{representable} in $\LLwot$ iff there is a finite term $\termone_\funone$ 
such that for every $\strone_1\in\pfininfs{\alpone}{\starone_1},\dots,
\strone_n\in\pfininfs{\alpone}{\starone_n}$ it holds that
$$
\ap{\termone_\funone}{\psemb{\strone_1}{\starone_1}\cdots\psemb{\strone_n}{\starone_n}}\redLLinflbl\psemb{\funone(\strone_1,\ldots,\strone_n)}{\startwo}
$$
if $\funone(\strone_1,\ldots,\strone_n)$ is defined, while $\ap{\termone_\funone}{\psemb{\strone_1}{\starone_1}\cdots\psemb{\strone_n}{\starone_n}}$
has no normal form otherwise. Notice the use of level-by-level reduction. Noticeably:
\begin{theorem}[Universality]
  The class of functions which are representable in $\LLwot$ coincides with the class of functions computable by
  Type-2 Turing machines.
\end{theorem}
\begin{proof}
This proof relies on a standard encoding of Turing machines into $\LLwot$. Rather than describing the encoding in
detail, we now give some observations, that together should convince the reader that the encoding is indeed
possible:
\begin{varitemize}
\item
  First of all, inductive and coinductive fixed point combinators are both available in $\LLwot$. Indeed, let
  $\termone_\bitone$ be the following term:
  $$
  \termone_\bitone=\ia{\varone}{\ia{\vartwo}{\ap{\vartwo}{\pam{\bitone}{(\ap{(\ap{\varone}{\im{\varone}})}{\im{\vartwo}})}}}}
  $$
  Then $\fpc_\bitone$ is just $\ap{\termone_\bitone}{\im{\termone_\bitone}}$. Observe that
  $\fpc_\bitone\im{\termone}\redLLinf\termone\pam{\bitone}{(\ap{\fpc_\bitone}{\im{\termone}})}$.
\item
  Moreover, observe that the encoding of free (co)algebras described above not only provides an elegant way to represent
  \emph{terms}, but also allows to very easily define efficient combinators for selection. For example, given the alphabet
  $\alpone=\{0,1\}$, the algebra $\fa{\alpone}$ of binary strings corresponds to the term
  $$
  \termone=\la{\varone}{\ia{\vartwo_0}{\ia{\vartwo_1}{\ia{\vartwo_\varepsilon}
        {\ap{\varone}{\im{\vartwo_0}\im{\vartwo_1}\im{\vartwo_\varepsilon}}}}}}.
  $$
  Please observe that
  \begin{align*}
  \termone\faemb{\alpone}{0(\strone)}\im{\termtwo_0}\im{\termtwo_1}\im{\termtwo_\varepsilon}&\redLLinf
    \termtwo_0\faemb{\alpone}{\strone}{};\\
  \termone\faemb{\alpone}{1(\strone)}\im{\termtwo_0}\im{\termtwo_1}\im{\termtwo_\varepsilon}&\redLLinf
    \termtwo_1\faemb{\alpone}{\strone}{};\\
  \termone\faemb{\alpone}{\varepsilon}\im{\termtwo_0}\im{\termtwo_1}\im{\termtwo_\varepsilon}&\redLLinf
    \termtwo_\varepsilon\faemb{\alpone}{\strone}.
  \end{align*}
  This can be generalised to an arbitrary (co)algebra.
\item
  Tuples can be represented easily as follows:
  $$
  \la{\varone}{\ap{\varone}{\im{\termone_1}\ldots\im{\termone_\natone}}}.
  $$
\item
  A configuration of a Type-2 Turing machine working on the alphabet $\alpone$, with states in the
  set $\stsone$, and having $\natone$ input tapes and $\nattwo$ working tapes can be seen
  as the following $\natone+3\nattwo+1$-tuple:
  $$
  (\strone_1,\ldots,\strone_\natone,\strtwo_1^l,\charone_1,\strtwo_1^r,\ldots,\strtwo_\nattwo^l,\charone_\nattwo,\strtwo_\nattwo^r,\stone)
  $$
  where $\strone_i$ is the not-yet-read portion of the $i$-th input tape,
  $\strtwo_i^l$ (respectively, $\strtwo_i^r$) is the portion of the $i$-th working tape to the left (respectively, right) 
  of the head, $\charone_i$ is the symbol currently under the $i$-th head, and $\stone$ is the current state.
  All these $\natone+3\nattwo+1$ objects can be seen as elements of appropriate (co)algebras:
  \begin{varitemize}
    \item
      $\strone_1,\ldots,\strone_\natone$ are (finite or infinite, depending on the underlying machine) strings;
    \item
      $\charone_1,\ldots,\charone_\nattwo,\stone$ are all elements of finite sets;
    \item
      $\strtwo_1^l,\strtwo_1^r,\ldots,\strtwo_\nattwo^l,\strtwo_\nattwo^r$ are finite strings;
  \end{varitemize}
  As a consequence, all of them can be encoded in Scott-style. Moreover, the availability of selectors
  and tuples makes it easy to write a term encoding the transition function of the encoded machine. Please notice that the 
  output tape is \emph{not} part of the configuration above. 
  If a character is produced in output (i.e. whenever the head of the output tape moves right), the rest
  of the computation will take place ``inside'' the encoding of a (possibly infinite) string.
\item
  One needs to be careful when handling final states: if the machine reaches a final state \emph{even if} 
  it is meant to produce infinite strings in output, the encoding $\lambda$-term should anyway diverge~\cite{Weihrauch00}.
\end{varitemize}
Putting the ingredients above together, one gets for every machine $\mdtone$ a term $\termone_\mdtone$ which
computes the same function as $\mdtone$. The fact that anything computable by a finite term $\termone$
is also computable by a Type-2 Turing machine is quite easy, since level-by-level reduction is effective and,
moreover, a normal form in the sense of level-by-level reduction is reached iff it is reached applying
surface reduction (in the sense of inductive boxes) \emph{at each depth}.
\end{proof}
\section{Taming Infinity: $\LLLwot$}
As we have seen in the last sections, $\LLwot$ is a very powerful
model: not only it is universal as a way to compute over streams, but
also comes with an extremely liberal infinitary dynamics for which,
unfortunately, confluence does not hold. If this paper finished here,
this work would then be rather inconclusive: $\LLwot$ suffers from the same
kind of defects which its nonlinear sibling $\Lwot$ has.

In this section, however, we will define a restriction of $\LLwot$,
called $\LLLwot$, which thanks to a careful management of boxes in the
style of light logics~\cite{Girard98IC,Danos03IC,Lafont04TCS}, allows
to keep infinity under control, and to get results which are
impossible to achieve in $\Lwot$.

Actually, the notion of a preterm remains unaltered. What changes is
how \emph{terms} are defined.  First of all, patterns are generalised
by two new productions $\patone::=\dm{\varone}\midd\am{\varone}$, which
make the notion of an environment slightly more general: it can now
contain variables in \emph{five} different forms. Judgments have the
usual shape, namely $\tjudg{\conone}{\termone}$ where $\conone$ is an
environment and $\termone$ is a term. Metavariables like $\liconone$
or $\licontwo$ stand for environments where the only allowed patterns
are either variables or the ones in the form $\im{\varone}$. Rules of
$\LLLwot$ are quite different than the ones of $\LLwot$, and can be
found in Figure~\ref{fig:lllwotwfr}.
\begin{figure*}
\begin{center}
\fbox{
\begin{minipage}{.97\textwidth}
  \begin{center}
  \vspace{6pt}
  {\footnotesize
    $$
    \infer[\LLvarl]
    {\tjudg{\dm{\lconone},\cm{\lcontwo},\am{\lconthree},\varone}{\varone}}
    {}
    \qquad
    \infer[\LLvard]
    {\tjudg{\dm{\lconone},\cm{\lcontwo},\am{\lconthree},\dm{\varone}}{\varone}}
    {}
    \qquad
    \infer[\LLvara]
    {\tjudg{\dm{\lconone},\cm{\lcontwo},\am{\lconthree},\am{\varone}}{\varone}}
    {}    
    $$
    $$
    \infer[\LLapp]
    {\tjudg{\liconone,\licontwo,\dm{\lconone},\cm{\lcontwo},\am{\lconthree}}{\ap{\termone}{\termtwo}}}
    {
      \tjudg{\liconone,\dm{\lconone},\cm{\lcontwo},\am{\lconthree}}{\termone}
      &
      \tjudg{\licontwo,\dm{\lconone},\cm{\lcontwo},\am{\lconthree}}{\termtwo}
    }
    \qquad
    \infer[\LLlaml]
    {\tjudg{\conone}{\la{\varone}{\termone}}}
    {\tjudg{\conone,\varone}{\termone}}
    \qquad
    \infer[\LLlami_1]
    {\tjudg{\conone}{\ia{\varone}{\termone}}}
    {\tjudg{\conone,\dm{\varone}}{\termone}}
    $$
    $$
    \infer[\LLlami_2]
    {\tjudg{\conone}{\ia{\varone}{\termone}}}
    {\tjudg{\conone,\im{\varone}}{\termone}}
    \qquad
    \infer[\LLlamc]
    {\tjudg{\conone}{\ca{\varone}{\termone}}}
    {\tjudg{\conone,\cm{\varone}}{\termone}}
    \qquad
    \infer[\LLlmi]
    {\tjudg{\dm{\lconone},\im{\lcontwo},\cm{\lconthree},\am{\lconfour}}{\im{\termone}}}
    {\tjudg{\lcontwo,\cm{\lconthree},\am{\lconfour}}{\termone}}
    \qquad
    \infer=[\LLlmc]
    {\tjudg{\dm{\lconone},\cm{\lcontwo},\am{\lconthree}}{\cm{\termone}}}
    {\tjudg{\am{\lcontwo},\am{\lconthree}}{\termone}}
    $$}
    \vspace{1pt}
  \end{center}
\end{minipage}}
\end{center}
  \caption{$\LLLwot$: Well-Formation Rules.}\label{fig:lllwotwfr}
\end{figure*}
The meaning well-formation rules induce on variable occurring in
environments is more complicated than for $\LLwot$.  Suppose that
$\tjudg{\conone}{\termone}$. Then:
\begin{varenumerate}
\item\label{enum:first}
  If $\varone\in\conone$ then, as usual, $\varone$ occurs once in
  $\termone$, and outside of any box;
\item\label{enum:second} 
  If $\dm{\varone}\in\conone$ then $\varone$
  can occur any number of times in $\termone$, but all these
  occurrences are in \emph{linear} position, i.e., outside the scope
  of \emph{any} box;
\item\label{enum:third}
  If $\im{\varone}\in\conone$ then $\varone$ occurs exactly once in
  $\termone$, and in the scope of exactly one (inductive) box;
\item\label{enum:fourth}
  If $\cm{\varone}\in\conone$, then $\varone$ occurs any number of
  times in $\termone$, with the only proviso that any such occurrence
  of $\varone$ must be in the scope of \emph{at least} one coinductive
  box.
\item\label{enum:fifth}
  Finally, if $\am{\varone}\in\conone$, then $\varone$ occurs any
  number of times in $\termone$, in any possible position.
\end{varenumerate}
Conditions \ref{enum:first}. to \ref{enum:third}. are reminiscent of
the ones of Lafont's soft linear logic.  Analogously,
Condition~\ref{enum:fourth}. is very much in the style of \FLL\ as
described by Danos and Joinet~\cite{Danos03IC}: $\csym$ is morally a
functor for which contraction, weakening and digging are available,
but which does not support dereliction.  We will come back to the
consequences of this exponential discipline below in this
section. Please observe that any variable $\varone$ marked as
$\dm{\varone}$ or $\im{\varone}$ cannot occur in the scope of
coinductive boxes. The pattern $\am{\varone}$ has only a merely
technical role.

If $\conone$ and $\contwo$ are environments, we write
$\conone\envleq\contwo$ iff $\contwo$ can be obtained from $\conone$
by replacing \emph{some} patterns in the form $\im{\varone}$ with
$\dm{\varone}$.  Well-formation is preserved by reduction and, as for
$\LLwot$, a proof of this fact requires a number of substitution
lemmas:
\begin{lemma}[Substitution Lemma, Linear Case]\label{lemma:lllsllinear}
If $\tjudg{\liconone,\varone,\dm{\lconone},\cm{\lcontwo},\am{\lconthree}}{\termone}$ and
$\tjudg{\licontwo,\dm{\lconone},\cm{\lcontwo},\am{\lconthree}}{\termtwo}$, then
$\tjudg{\liconone,\licontwo,\dm{\lconone},\cm{\lcontwo},\am{\lconthree}}{\sbst{\termone}{\varone}{\termtwo}}$.
\end{lemma}
\begin{lemma}[Substitution Lemma, First Inductive Case]\label{lemma:lllslinductiveI}
If $\tjudg{\liconone,\dm{\lconone},\dm{\varone},\cm{\lcontwo},\am{\lconthree}}{\termone}$ and
$\tjudg{\lconfour,\cm{\lcontwo},\am{\lconthree}}{\termtwo}$, then
$\tjudg{\liconone,\dm{\lconone},\dm{\lconfour},\cm{\lcontwo},\am{\lconthree}}{\sbst{\termone}{\varone}{\termtwo}}$.
\end{lemma}
\begin{lemma}[Substitution Lemma, Second Inductive Case]\label{lemma:lllslinductiveII}
If $\tjudg{\liconone,\dm{\lconone},\im{\varone},\cm{\lcontwo},\am{\lconthree}}{\termone}$ and
$\tjudg{\lconfour,\cm{\lcontwo},\am{\lconthree}}{\termtwo}$, then
$\tjudg{\liconone,\dm{\lconone},\im{\lconfour},\cm{\lcontwo},\am{\lconthree}}{\sbst{\termone}{\varone}{\termtwo}}$.
\end{lemma}
\begin{lemma}[Substitution Lemma, Coinductive Case]\label{lemma:lllslcoinductive}
If $\tjudg{\liconone,\dm{\lconone},\cm{\lcontwo},\cm{\varone},\am{\lconthree}}{\termone}$ and
$\tjudg{\am{\lcontwo},\am{\lconthree}}{\termtwo}$, then
$\tjudg{\liconone,\dm{\lconone},\cm{\lcontwo},\am{\lconthree}}{\sbst{\termone}{\varone}{\termtwo}}$.
\end{lemma}
\begin{lemma}[Substitution Lemma, Arbitrary Case]\label{lemma:lllslarbitrary}
If $\tjudg{\liconone,\dm{\lconone},\cm{\lcontwo},\am{\lconthree},\am{\varone}}{\termone}$ and
$\tjudg{\am{\lconthree}}{\termtwo}$, then
$\tjudg{\liconone,\dm{\lconone},\cm{\lcontwo},\am{\lconthree}}{\sbst{\termone}{\varone}{\termtwo}}$.
\end{lemma}
Altogether, the lemmas above imply
\begin{proposition}[Well-Formedness is Preseved by Reduction]
  If $\tjudg{\conone}{\termone}$ and $\termone\redLLwod\termtwo$, then $\tjudg{\contwo}{\termtwo}$
  where $\conone\envleq\contwo$.
\end{proposition}
Please observe that the underlying environment can indeed change
during reduction, but in a very peculiar way: variables occurring in a
$\isym$-pattern can later move to a $\dsym$-pattern.

Classes $\tms{\LLLwot}$ and $\tms{\LLLwot}(\conone)$ (where $\conone$
is an environment) are defined in the natural way, as in $\LLwot$.
\subsection{The Fundamental Lemma}
It is now time to show \emph{why} $\LLLwot$ is a computationally
well-behaved object. 
In this section we will prove a crucial result,
namely that reduction is strongly normalising \emph{at each
  depth}.
Before embarking on the proof of this result, let us spend
some time to understand why this is the case, giving some necessary
definitions along the way.

For any term $\termone\in\tms{\LLLwot}$ let us define the \emph{size}
$\psize{\termone}{\natone}$ of $\termone$ \emph{at depth} $\natone$ as the
number of occurrences of any symbol at depth $\natone$ inside
$\termone$. Observe that $\psize{\termone}{\natone}$ is well-defined
\emph{only} because $\termone$ is assumed to be a term and not just a
preterm. Formally, $\psize{\termone}{\natone}$ is any
  natural number satisfying the equations in
  Figure~\ref{fig:psizetms}, and the following result holds:
  \begin{figure}
    \begin{center}
\fbox{
\begin{minipage}{.445\textwidth}
\begin{center}
\begin{align*}
\psize{\varone}{0}&=1;\\   
\psize{\im{\termone}}{0}&=\psize{\termone}{0}+1;\\
\psize{\cm{\termone}}{0}&=0;\\
\psize{\ap{\termone}{\termtwo}}{0}&=\psize{\termone}{0}+\psize{\termtwo}{0}+1;\\
\psize{\la{\varone}{\termone}}{0}&=\psize{\ia{\varone}{\termone}}{0}\\
  &=\psize{\ca{\varone}{\termone}}{0}=\psize{\termone}{0}+1;
\end{align*}
\begin{align*}
\psize{\varone}{\nattwo+1}&=0;\\
\psize{\im{\termone}}{\nattwo+1}&=\psize{\termone}{\nattwo+1};\\
\psize{\cm{\termone}}{\nattwo+1}&=\psize{\termone}{\nattwo};\\
\psize{\ap{\termone}{\termtwo}}{\nattwo+1}&=\psize{\termone}{\nattwo+1}+\psize{\termtwo}{\nattwo+1};\\
\psize{\la{\varone}{\termone}}{\nattwo+1}&=\psize{\ia{\varone}{\termone}}{\nattwo+1}\\
  &=\psize{\ca{\varone}{\termone}}{\nattwo+1}=\psize{\termone}{\nattwo+1}.
\end{align*}
\end{center}
\vspace{2pt}
\end{minipage}}
\end{center}
\caption{Parametrised Sizes of Preterms: Equations.}\label{fig:psizetms}
\end{figure}
\begin{lemma}\label{lemma:psize}
For every term $\termone$ and for every natural number $\nattwo\in\NN$ there is a unique
natural number $\natone$ such that $\psize{\termone}{\nattwo}=\natone$.
\end{lemma}
\begin{proof}
The fact that for each $\termone$ and for each $\nattwo$ there is \emph{one} natural number
satisfying the equations in Figure~\ref{fig:psizetms} can be proved by induction on $\nattwo$:
\begin{varitemize}
\item
  If $\nattwo=0$, then since $\termone$ is a term, $\tjudg{\conone}{\termone}$ is an element of the
  set $\indrul{\LLLwot}(\coindrul{\LLLwot}(\cd{\LLLwot}))$. Then, let us perform another induction
  on the (finite) number of inductive rules used to obtain $\tjudg{\conone}{\termone}$
  from something in $\coindrul{\LLLwot}(\cd{\LLLwot})$:
  \begin{varitemize}
  \item
    If the last rule is $\LLvarl$, $\LLvard$ or $\LLvara$, then $\psize{\termone}{0}=1$ by
    definition;
  \item
    If the last rule is $\LLapp$, then $\termone=\ap{\termtwo}{\termthree}$,
    there are $\psize{\termtwo}{0}$ and $\psize{\termthree}{0}$,
    and $\psize{\termone}{0}=\psize{\termtwo}{0}+\psize{\termthree}{0}+1$;
  \item
    If the last rule is either $\LLlaml$, $\LLlami_1$, $\LLlami_2$, $\LLlamc$ or
    $\LLlmi$, then we can proceed like in the previous case, by the inductive hypothesis;
  \item
    If the last rule is $\LLlmc$, then $\psize{\termone}{0}=0$ by definition.
  \end{varitemize}
\item
  If $\nattwo\geq 1$, then again, since $\termone$ is a term, $\tjudg{\conone}{\termone}$ is an element of the
  set $\indrul{\LLLwot}(\coindrul{\LLLwot}(\cd{\LLLwot}))$. Then, let us perform an induction
  on the (finite) number of inductive rules used to get $\tjudg{\conone}{\termone}$
  from something in $\coindrul{\LLLwot}(\cd{\LLLwot})$. The only interesting
  case is $\termone=\cm{\termtwo}$, since in all the other cases we can proceed exactly
  as in the case $\nattwo=0$. From the fact that $\tjudg{\conone}{\termone}\in\indrul{\LLLwot}(\coindrul{\LLLwot}(\cd{\LLLwot}))$,
  it follows that there is $\contwo$ such that $\tjudg{\contwo}{\termtwo}\in\cd{\LLLwot}$, i.e. $\termtwo$
  itself is a term. But then we can apply the inductive hypothesis and obtain that
  $\psize{\termtwo}{\nattwo-1}$ exists. It is now clear that $\psize{\termone}{\nattwo}=\psize{\termtwo}{\nattwo-1}$
  exists.
\end{varitemize}
As for uniqueness, it can be proved by observing that the equations from figure~\ref{fig:psizetms} can be oriented
so as to get a confluent rewrite system for which, then, the Church-Rosser property holds. This concludes the proof.
\end{proof}

Now, suppose that a term $\termone$ is such that
$\termone\redLL{\strnat{\natone}}\termfour$. The term $\termone$, then, must be in the
form $\actx{\pctxone{\strnat{\natone}}}{\termtwo}$ where $\termtwo$ is a redex
whose reduct is $\termthree$, and $\termfour$ is just
$\actx{\pctxone{\strnat{\natone}}}{\termthree}$. The question is: how does any
$\psize{\actx{\pctxone{\strnat{\natone}}}{\termtwo}}{\nattwo}$ relate to the
corresponding $\psize{\actx{\pctxone{\strnat{\natone}}}{\termthree}}{\nattwo}$?
Some interesting observations follow:
\begin{varitemize}
\item
  If $\nattwo<\natone$ then
  $\psize{\actx{\pctxone{\strnat{\natone}}}{\termthree}}{\nattwo}$
  equals
  $\psize{\actx{\pctxone{\strnat{\natone}}}{\termtwo}}{\nattwo}$,
  since reduction does not affect the size at lower levels;
\item
  If $\nattwo>\natone$ then of course
  $\psize{\actx{\pctxone{\strnat{\natone}}}{\termthree}}{\nattwo}$ can
  be much bigger than
  $\psize{\actx{\pctxone{\strnat{\natone}}}{\termtwo}}{\nattwo}$,
  simply because symbol occurrences at depth $\nattwo$ can be
  duplicated as an effect of substitution;
\item
  Finally, if $\nattwo=\natone$, then
  $\natthree=\psize{\actx{\pctxone{\strnat{\natone}}}{\termthree}}{\nattwo}$
  can again be bigger than
  $\natfour=\psize{\actx{\pctxone{\strnat{\natone}}}{\termtwo}}{\nattwo}$,
  but in a very controlled way. More specifically,
  \begin{varitemize}
    \item
      if $\termtwo$ is a \emph{linear} redex, then $\natfour<\natthree$
      because the function body has exactly one free occurrence of the
      bound variable;
    \item
      if $\termtwo$ is an \emph{inductive} redex, then $\natfour$ can indeed
      by bigger than $\natthree$, but in that case the involved
      inductive box has disappeared.
    \item
      if $\termtwo$ is a \emph{coinductive} redex, then $\natfour<\natthree$
      because the involved coinductive box can actually be copied many
      times, but all the various copies will be found at depths
      strictly bigger than $\natone=\nattwo$.
  \end{varitemize}
\end{varitemize}
The informal argument above can be formalised by way of an appropriate
notion of weight, generalising the argument by Lafont~\cite{Lafont04TCS}
to the more general setting we work in here.

Given $\natone,\nattwo\in\NN$ and a term $\termone$, the
\emph{$\natone$-weight} $\wei{\natone}{\nattwo}{\termone}$ of
$\termone$ at depth $\nattwo$ is any natural number satisfying the
rules from Figure~\ref{fig:pwtms}.
\begin{figure}
  \begin{center}
    \fbox{
      \begin{minipage}{.445\textwidth}
        \begin{center}
          \begin{align*}
            \wei{\natone}{0}{\varone}&=1;\\   
            \wei{\natone}{0}{\im{\termone}}&=\natone\cdot\wei{\natone}{0}{\termone};\\
            \wei{\natone}{0}{\cm{\termone}}&=0;\\
            \wei{\natone}{0}{\ap{\termone}{\termtwo}}&=\wei{\natone}{0}{\termone}+\wei{\natone}{0}{\termtwo};\\
            \wei{\natone}{0}{\la{\varone}{\termone}}&=\wei{\natone}{0}{\ia{\varone}{\termone}}\\
            &=\wei{\natone}{0}{\ca{\varone}{\termone}}=\wei{\natone}{0}{\termone}+1;
          \end{align*}
          \begin{align*}
            \wei{\natone}{\nattwo+1}{\varone}&=0;\\
            \wei{\natone}{\nattwo+1}{\im{\termone}}&=\wei{\natone}{\nattwo+1}{\termone};\\
            \wei{\natone}{\nattwo+1}{\cm{\termone}}&=\wei{\natone}{\nattwo}{\termone};\\
            \wei{\natone}{\nattwo+1}{\ap{\termone}{\termtwo}}&=\wei{\natone}{\nattwo+1}{\termone}+\wei{\natone}{\nattwo+1}{\termtwo};\\
            \wei{\natone}{\nattwo+1}{\la{\varone}{\termone}}&=\wei{\natone}{\nattwo+1}{\ia{\varone}{\termone}}\\
            &=\wei{\natone}{\nattwo+1}{\ca{\varone}{\termone}}=\wei{\natone}{\nattwo+1}{\termone}.
          \end{align*}
        \end{center}
        \vspace{2pt}
    \end{minipage}}
  \end{center}
  \caption{Parametrised Weights of Preterms: Equations.}\label{fig:pwtms}
\end{figure}
\begin{lemma}
  For every term $\termone$ and for every natural numbers
  $\natone,\nattwo\in\NN$, there is a unique natural number $\natthree$
  such that $\wei{\natone}{\nattwo}{\termone}=\natthree$.
\end{lemma}
\begin{proof}
  This can be proved to hold in exactly the same way as we did for the
  size in Lemma~\ref{lemma:psize}.
\end{proof}
Similarly, one can define the duplicability factor of $\termone$ at
depth $\nattwo$, $\df{\nattwo}{\termone}$: take the rules in
Figure~\ref{fig:dftms} and prove they uniquely define a natural number
for every term, in the same way as we have just done for the weight
($\nfo{\varone}{\termone}$ is the number of free occurrences of
$\varone$ in the term $\termone$, itself a well-defined concept when
$\termone$ is a term).
\begin{figure}
  \begin{center}
\fbox{
\begin{minipage}{.445\textwidth}
\begin{center}
\begin{align*}
\df{0}{\varone}&=1;\\   
\df{0}{\im{\termone}}&=\df{0}{\termone};\\
\df{0}{\cm{\termone}}&=1;\\
\df{0}{\ap{\termone}{\termtwo}}&=\max\{\df{0}{\termone},\df{0}{\termtwo}\};\\
\df{0}{\la{\varone}{\termone}}&=\df{0}{\ca{\varone}{\termone}}=\df{0}{\termone};\\
\df{0}{\ia{\varone}{\termone}}&=\max\{\nfo{\varone}{\termone},\df{0}{\termone}\}.
\end{align*}
\begin{align*}
\df{\nattwo+1}{\varone}&=1;\\
\df{\nattwo+1}{\im{\termone}}&=\df{\nattwo+1}{\termone};\\
\df{\nattwo+1}{\cm{\termone}}&=\df{\nattwo}{\termone};\\
\df{\nattwo+1}{\ap{\termone}{\termtwo}}&=\max\{\df{\nattwo+1}{\termone},\df{\nattwo+1}{\termtwo}\};\\
\df{\nattwo+1}{\la{\varone}{\termone}}&=\df{\nattwo+1}{\ia{\varone}{\termone}}\\
  &=\df{\nattwo+1}{\ca{\varone}{\termone}}=\df{\nattwo+1}{\termone}.
\end{align*}
\end{center}
\vspace{2pt}
\end{minipage}}
\end{center}
\caption{Parametrised Duplicability Factor of Preterms: Equations.}\label{fig:dftms}
\end{figure}
Given a term $\termone$, the weight of $\termone$ at depth $\natone$ is simply
$\twei{\natone}{\termone}=\wei{\df{\natone}{\termone}}{\natone}{\termone}$.

The calculus $\LLLwot$ is designed in such a way that the duplicability factor never increases:
\begin{lemma}\label{lemma:dfdni}
If $\termone\in\tms{\LLLwot}$ and $\termone\redLLwod\termtwo$, then
$\df{\nattwo}{\termone}\geq\df{\nattwo}{\termtwo}$ for every $\nattwo$.
\end{lemma}
\begin{proof}
A formal proof could be given. We prefer, however, to give a more intuitive one here.
Observe that:
\begin{varitemize}
\item
  If $\tjudg{\conone,\varone}{\termone}$ or $\tjudg{\conone,\im{\varone}}{\termone}$, then the variable $\varone$ occurs 
  free exactly once in $\termone$, in the first case outside the scope of any box, in the second case in the scope of
  exactly one inductive box.
\item
  If $\tjudg{\conone,\dm{\varone}}{\termone}$, then $\varone$ occurs free more than once in $\termone$, all the occurrences
  being outside the scope of any box.
\item
  The duplicability factor at level $\natone$ of $\termone$ is nothing more than the maximum, over all abstractions $\ia{\varone}{\termtwo}$
  at level $\natone$ in $\termone$, of the number of free occurrences of $\varone$ in $\termtwo$. Observe that by the well-formation rules
  in Figure~\ref{fig:lllwotwfr}, the variable $\varone$ must be marked as $\dm{\varone}$ or as $\im{\varone}$ for any such $\termtwo$.
  If it is marked as $\im{\varone}$, however, it occurs once in $\termtwo$.
\item
  Now, consider the substitution lemmas~\ref{lemma:lllsllinear}, \ref{lemma:lllslinductiveI}, \ref{lemma:lllslinductiveII},
  \ref{lemma:lllslcoinductive}, and \ref{lemma:lllslarbitrary}. In all the five cases, one realises that:
  \begin{varenumerate}
    \item
      for every $\natone$, $\df{\natone}{\sbst{\termone}{\varone}{\termtwo}}\leq\max\{\df{\natone}{\termone},\df{\natone}{\termtwo}\}$,
      because every abstraction occurring in $\sbst{\termone}{\varone}{\termtwo}$ also occurs in either $\termone$ or $\termtwo$, and
      substitution is capture-avoiding. 
    \item
      If in the judgment $\tjudg{\conone}{\sbst{\termone}{\varone}{\termtwo}}$ one gets as a result of the substitution lemma there
      is $\dm{\vartwo}\in\conone$, then $\nfo{\vartwo}{\sbst{\termone}{\varone}{\termtwo}}$ cannot be too big: there must be
      some $\varthree$ such that $\varthree$ is marked as $\dm{\varthree}$ in one (or both) of the provable judgments existing by hypothesis,
      but $\nfo{\varthree}{\termone}+\nfo{\varthree}{\termtwo}$ majorises $\nfo{\vartwo}{\sbst{\termone}{\varone}{\termtwo}}$. Why?
      Simply because the only case in which $\nfo{\vartwo}{\sbst{\termone}{\varone}{\termtwo}}$ could potentially grow bigger is the one in
      which $\varone$ is marked as $\dm{\varone}$ in the judgment for $\termone$. In that case, however, the variables which are
      free in $\termtwo$ are all linear (or marked as $\cm{\varthree}$ or $\am{\varthree}$).
  \end{varenumerate}
\end{varitemize}
This concludes the proof.
\end{proof}
Moreover, and this is the crucial point, $\twei{\natone}{\termone}$ is
guaranteed to strictly decrease whenever
$\termone\redLL{\strnat{\natone}}\termtwo$:
\begin{lemma}\label{lemma:weightdecr}
Suppose that $\termone\in\tms{\LLLwot}$ and that $\termone\redLL{\strnat{\natone}}\termtwo$. Then
$\twei{\natone}{\termone}>\twei{\natone}{\termtwo}$. Moreover, $\twei{\nattwo}{\termone}=\twei{\nattwo}{\termtwo}$
whenever $\nattwo<\natone$.
\end{lemma}
\begin{proof}
  We first of all need to prove the following variations on the substitution lemmas:
  \begin{varenumerate}
  \item\label{point:sublemmaI}
    If $\tjudg{\liconone,\varone,\dm{\lconone},\cm{\lcontwo},\am{\lconthree}}{\termone}$ and
    $\tjudg{\licontwo,\dm{\lconone},\cm{\lcontwo},\am{\lconthree}}{\termtwo}$, then
    for every $\natone\geq\max\{\df{0}{\termone},\df{0}{\termtwo}\}$ it holds
    that $\wei{\natone}{0}{\sbst{\termone}{\varone}{\termtwo}}\leq
    \wei{\natone}{0}{\termone}+\wei{\natone}{0}{\termtwo}$.
  \item\label{point:sublemmaII}
    If $\tjudg{\liconone,\dm{\lconone},\dm{\varone},\cm{\lcontwo},\am{\lconthree}}{\termone}$ and
    $\tjudg{\lconfour,\cm{\lcontwo},\am{\lconthree}}{\termtwo}$, then
    for every $\natone\geq\max\{\df{0}{\termone},\df{0}{\termtwo}\}$ it holds
    that $\wei{\natone}{0}{\sbst{\termone}{\varone}{\termtwo}}\leq
    \wei{\natone}{0}{\termone}+\nfo{\varone}{\termone}\cdot\wei{\natone}{0}{\termtwo}$.
  \item\label{point:sublemmaIII}
    If $\tjudg{\liconone,\dm{\lconone},\im{\varone},\cm{\lcontwo},\am{\lconthree}}{\termone}$ and
    $\tjudg{\lconfour,\cm{\lcontwo},\am{\lconthree}}{\termtwo}$, then
    for every $\natone\geq\max\{\df{0}{\termone},\df{0}{\termtwo}\}$ it holds
    that $\wei{\natone}{0}{\sbst{\termone}{\varone}{\termtwo}}\leq
    \wei{\natone}{0}{\termone}+\natone\cdot\wei{\natone}{0}{\termtwo}$.
   \item\label{point:sublemmaIV}
    If $\tjudg{\liconone,\dm{\lconone},\cm{\lcontwo},\am{\lconthree},\am{\varone}}{\termone}$ and
    $\tjudg{\am{\lconthree}}{\termtwo}$, then
    for every $\natone\geq\max\{\df{0}{\termone},\df{0}{\termtwo}\}$ it holds
    that $\wei{\natone}{0}{\sbst{\termone}{\varone}{\termtwo}}\leq
    \wei{\natone}{0}{\termone}$.
   \item\label{point:sublemmaV}
    If $\tjudg{\liconone,\dm{\lconone},\cm{\lcontwo},\cm{\varone},\am{\lconthree}}{\termone}$ and
    $\tjudg{\am{\lcontwo},\am{\lconthree}}{\termtwo}$, then
    for every $\natone\geq\max\{\df{0}{\termone},\df{0}{\termtwo}\}$ it holds
    that $\wei{\natone}{0}{\sbst{\termone}{\varone}{\termtwo}}\leq
    \wei{\natone}{0}{\termone}$.
  \end{varenumerate}
  All the statements above can be proved, as usual, by induction on the (finite) number of
  inductive well-formation rules which are necessary to prove the judgment about $\termone$ from
  something in $\coindrul{\LLLwot}(\cd{\LLLwot})$. As an example, let us consider some
  inductive cases on Point~\ref{point:sublemmaII}., which is one of the most interesting:
  \begin{varitemize}
  \item
    If $\termone$ is proved well-formed by
    $$
    \infer[\LLvard]
    {\tjudg{\dm{\lconone},\cm{\lcontwo},\am{\lconthree},\dm{\varone}}{\varone}}
    {}
    $$
    then $\sbst{\termone}{\varone}{\termtwo}=\termtwo$ and
    \begin{align*}
    \wei{\natone}{0}{\sbst{\termone}{\varone}{\termtwo}}&=\wei{\natone}{0}{\termtwo}\\
       &=1\cdot\wei{\natone}{0}{\termtwo}=\nfo{\varone}{\termone}\cdot\wei{\natone}{0}{\termtwo}\\
       &\leq\wei{\natone}{0}{\termone}+\nfo{\varone}{\termone}\cdot\wei{\natone}{0}{\termtwo}.
    \end{align*}
  \item
    If $\termone$ is proved well-formed by
    $$
    \infer[\LLapp]
    {\tjudg{\liconone,\licontwo,\dm{\lconone},\dm{\varone},\cm{\lcontwo},\am{\lconthree}}{\ap{\termthree}{\termfour}}}
    {
      \tjudg{\liconone,\dm{\lconone},\dm{\varone},\cm{\lcontwo},\am{\lconthree}}{\termthree}
      &
      \tjudg{\licontwo,\dm{\lconone},\dm{\varone},\cm{\lcontwo},\am{\lconthree}}{\termfour}
    }
    $$
    then $\sbst{\termone}{\varone}{\termtwo}=\ap{(\sbst{\termthree}{\varone}{\termtwo})}{(\sbst{\termfour}{\varone}{\termtwo})}$
    and, by inductive hypothesis, we have
    \begin{align*}
      \wei{\natone}{0}{\sbst{\termthree}{\varone}{\termtwo}}&\leq\wei{\natone}{0}{\termthree}+\nfo{\varone}{\termthree}\cdot\wei{\natone}{0}{\termtwo};\\
      \wei{\natone}{0}{\sbst{\termfour}{\varone}{\termtwo}}&\leq\wei{\natone}{0}{\termfour}+\nfo{\varone}{\termfour}\cdot\wei{\natone}{0}{\termtwo}.
    \end{align*}
    But then:
    \begin{align*}
      \wei{\natone}{0}{\sbst{\termone}{\varone}{\termtwo}}&=\wei{\natone}{0}{\sbst{\termthree}{\varone}{\termtwo}}+\wei{\natone}{0}{\sbst{\termfour}{\varone}{\termtwo}}\\
         &\leq(\wei{\natone}{0}{\termthree}+\nfo{\varone}{\termthree}\cdot\wei{\natone}{0}{\termtwo})\\
         &\quad+(\wei{\natone}{0}{\termfour}+\nfo{\varone}{\termfour}\cdot\wei{\natone}{0}{\termtwo})\\
         &=(\wei{\natone}{0}{\termthree}+\wei{\natone}{0}{\termfour})\\
         &\quad+(\nfo{\varone}{\termthree}+\nfo{\varone}{\termfour})\cdot\wei{\natone}{0}{\termtwo}\\
         &=\wei{\natone}{0}{\termone}+\nfo{\varone}{\termone}\cdot\wei{\natone}{0}{\termtwo}.
    \end{align*}
  \item
    If $\termone$ is proved well-formed by
    $$
    \infer[\LLlmi]
    {\tjudg{\dm{\lconone},\dm{\varone},\im{\lcontwo},\cm{\lconthree},\am{\lconfour}}{\im{\termone}}}
    {\tjudg{\lcontwo,\cm{\lconthree},\am{\lconfour}}{\termone}}
    $$
    then $\varone$ does not occur free in $\termone$ and, as a consequence $\sbst{\termone}{\varone}{\termtwo}=\termone$. The thesis easily follows.
  \end{varitemize}
With the five lemmas above in our hands, it is possible to prove that if $\termone\redLLbas\termtwo$, then
$\twei{0}{\termone}>\twei{0}{\termtwo}$. Let's proceed by cases depending on how $\termone\redLLbas\termtwo$ is derived:
\begin{varitemize}
\item
  If $\termone=\ap{(\la{\varone}{\termthree})}{\termfour}$, then
  $\tjudg{\liconone,\varone,\dm{\lconone},\cm{\lcontwo},\am{\lconthree}}{\termthree}$
  and $\tjudg{\licontwo,\dm{\lconone},\cm{\lcontwo},\am{\lconthree}}{\termfour}$. We can apply Point \ref{point:sublemmaI}.
  (and Lemma~\ref{lemma:dfdni}) obtaining
  \begin{align*}
    \twei{0}{\termone}&=\wei{\df{0}{\termone}}{0}{\termone}=\wei{\df{0}{\termone}}{0}{\termthree}+\wei{\df{0}{\termone}}{0}{\termfour}+1\\
      &>\wei{\df{0}{\termone}}{0}{\termthree}+\wei{\df{0}{\termone}}{0}{\termfour}\geq\wei{\df{0}{\termone}}{0}{\sbst{\termthree}{\varone}{\termfour}}\\
      &=\wei{\df{0}{\termone}}{0}{\termtwo}\geq\wei{\df{0}{\termtwo}}{0}{\termtwo}=\twei{0}{\termtwo}.
  \end{align*}
\item
  If $\termone=\ap{(\ia{\varone}{\termthree})}{\im{\termfour}}$, then we can distinguish two sub-cases:
  \begin{varitemize}
  \item
    If $\tjudg{\liconone,\dm{\varone},\dm{\lconone},\cm{\lcontwo},\am{\lconthree}}{\termthree}$
    and $\tjudg{\lconfour,\cm{\lcontwo},\am{\lconthree}}{\termfour}$, then we can apply Point \ref{point:sublemmaII}.
    (and Lemma~\ref{lemma:dfdni}) obtaining
    \begin{align*}
      \twei{0}{\termone}&=\wei{\df{0}{\termone}}{0}{\termone}\\
      &=\wei{\df{0}{\termone}}{0}{\termthree}+\df{0}{\termone}\cdot\wei{\df{0}{\termone}}{0}{\termfour}+1\\
      &>\wei{\df{0}{\termone}}{0}{\termthree}+\df{0}{\termone}\cdot\wei{\df{0}{\termone}}{0}{\termfour}\\
      &\geq\wei{\df{0}{\termone}}{0}{\termthree}+\nfo{\varone}{\termthree}\cdot\wei{\df{0}{\termone}}{0}{\termfour}\\
      &\geq\wei{\df{0}{\termone}}{0}{\sbst{\termthree}{\varone}{\termfour}}\\
      &=\wei{\df{0}{\termone}}{0}{\termtwo}\geq\wei{\df{0}{\termtwo}}{0}{\termtwo}\\
      &=\twei{0}{\termtwo}.
    \end{align*}
  \item
    If $\tjudg{\liconone,\im{\varone},\dm{\lconone},\cm{\lcontwo},\am{\lconthree}}{\termthree}$
    and $\tjudg{\lconfour,\cm{\lcontwo},\am{\lconthree}}{\termfour}$, then we can apply Point \ref{point:sublemmaIII}.
    (and Lemma~\ref{lemma:dfdni}) obtaining
    \begin{align*}
      \twei{0}{\termone}&=\wei{\df{0}{\termone}}{0}{\termone}\\
      &=\wei{\df{0}{\termone}}{0}{\termthree}+\df{0}{\termone}\cdot\wei{\df{0}{\termone}}{0}{\termfour}+1\\
      &>\wei{\df{0}{\termone}}{0}{\termthree}+\df{0}{\termone}\cdot\wei{\df{0}{\termone}}{0}{\termfour}\\
      &\geq\wei{\df{0}{\termone}}{0}{\sbst{\termthree}{\varone}{\termfour}}\\
      &=\wei{\df{0}{\termone}}{0}{\termtwo}\geq\wei{\df{0}{\termtwo}}{0}{\termtwo}=\twei{0}{\termtwo}.
    \end{align*}
  \end{varitemize}
\item
  If $\termone=\ap{(\ca{\varone}{\termthree})}{\cm{\termfour}}$, then
  $\tjudg{\liconone,\dm{\lconone},\cm{\varone},\cm{\lcontwo},\am{\lconthree}}{\termthree}$ and $\tjudg{\cm{\lcontwo},\am{\lconthree}}{\termfour}$.
  We can apply Point \ref{point:sublemmaV}. (and Lemma~\ref{lemma:dfdni}) obtaining
  \begin{align*}
    \twei{0}{\termone}&=\wei{\df{0}{\termone}}{0}{\termone}=\wei{\df{0}{\termone}}{0}{\termthree}+1\\
      &>\wei{\df{0}{\termone}}{0}{\termthree}\geq\wei{\df{0}{\termone}}{0}{\sbst{\termthree}{\varone}{\termfour}}\\
      &=\wei{\df{0}{\termone}}{0}{\termtwo}\geq\wei{\df{0}{\termtwo}}{0}{\termtwo}=\twei{0}{\termtwo}.
  \end{align*}
\end{varitemize}
Now, remember that $\termone\redLL{\strnat{\natone}}\termtwo$ iff there are a $\natone$-context $\pctxone{\strnat{\natone}}$ and two terms
$\termthree$ and $\termfour$ such that $\termthree\redLLbas\termfour$, $\termone=\actx{\pctxone{\strnat{\natone}}}{\termthree}$,
and $\termtwo=\actx{\pctxone{\strnat{\natone}}}{\termfour}$. The thesis can be proved by induction on $\pctxone{\strnat{\natone}}$.
\end{proof}
The Fundamental Lemma easily follows:
\begin{proposition}[Fundamental Lemma]
For every natural number $n\in\NN$, the relation $\redLL{\strnat{\natone}}$ is strongly normalising.
\end{proposition}
\subsection{Normalisation and Confluence}
The Fundamental Lemma has a number of interesting consequences, which
make the dynamics of $\LLLwot$ definitely better-behaved than that of
$\LLwot$. The first such result we give is a Weak-Normalisation
Theorem:
\begin{theorem}[Normalisation]
  For every term $\termone\in\tms{\LLLwot}$ there is a normal form $\termtwo\in\tms{\LLLwot}$ 
  such that $\termone\redLLinf\termtwo$.
\end{theorem}
\begin{proof}
This is an immediate consequence of the Fundamental Lemma: for every
term $\termone$, first reduce (by $\redLLwod$) all redexes
at depth $0$, obtaining $\termtwo$, and then normalise all
the subterms of $\termtwo$ at depth $1$ (by $\redLLnext$).
Conclude by observing that reduction at higher depths does not
influence lower depths.
\end{proof}
The way the two modalities interact in $\LLLwot$ has effects
which go beyond normalisation. More specifically, the two relations
$\redLLwod$ and $\redLLnext$ do not interfere like in $\LLwot$,
and as a consequence, we get a Confluence Theorem:
\begin{theorem}[Strong Confluence]\label{theo:conf}
If $\termone\in\tms{\LLLwot}$, $\termone\redLLinf\termtwo$
and $\termone\redLLinf\termthree$, then there is
$\termfour\in\LLLwot$ such that
$\termtwo\redLLinf\termfour$ and
$\termthree\redLLinf\termfour$.
\end{theorem}
The path to confluence requires some auxiliary lemmas:
\begin{lemma}\label{lemma:commam}
If
$\tjudg{\liconone,\dm{\lconone},\cm{\lcontwo},\am{\varone},\am{\lconthree}}{\termone}$
and $\tjudg{\am{\lconthree}}{\termtwo}$,
$\termone\redLLinf\termthree$, and $\termtwo\redLLinf\termfour$, then
$\sbst{\termone}{\varone}{\termtwo}\redLLinf\sbst{\termthree}{\varone}{\termfour}$.
\end{lemma}
\begin{proof}
This is a coinduction on $\termone$.
\end{proof}
\begin{lemma}\label{lemma:commcm}
If
$\tjudg{\liconone,\dm{\lconone},\cm{\lcontwo},\cm{\varone},\am{\lconthree}}{\termone}$
and $\tjudg{\am{\lcontwo},\am{\lconthree}}{\termtwo}$, 
$\termone\redLLnext\termthree$, and $\termtwo\redLLinf\termfour$, then
$\sbst{\termone}{\varone}{\termtwo}\redLLnext\sbst{\termthree}{\varone}{\termfour}$.
\end{lemma}
\begin{proof}
This is again a coinduction on the structure of $\termone$, exploiting 
Lemma~\ref{lemma:commam}
\end{proof}
\begin{lemma}[Noninterference]\label{lemma:noninterference}
If $\termone\in\tms{\LLLwot}$, $\termone\redLL{\strnat{0}}\termtwo$
and $\termone\redLLnext\termthree$, then there is
$\termfour\in\LLLwot$ such that
$\termtwo\redLLnext\termfour$ and
$\termthree\redLL{\strnat{0}}\termfour$.
\end{lemma}
\begin{proof}
By coinduction on the structure of $\termone$. Some interesting cases:
\begin{varitemize}
\item
  If $\termone=\termfive\termsix$ and $\termfive\redLL{\strnat{0}}\termseven$,
  then $\termtwo=\termseven\termsix$. By definition, 
  $\termfive\redLLnext\termeight$ and $\termsix\redLLnext\termnine$,
  where $\termthree=\termeight\termnine$. By induction hypothesis,
  there is $\termten$ such that $\termseven\redLLnext\termten$
  and $\termeight\redLL{\strnat{0}}\termten$. The term we are looking for,
  then, is just $\termfour=\termten\termnine$. Indeed,
  $\termtwo=\termseven\termsix\redLLnext\termten\termnine$ and,
  other other hand, $\termthree=\termeight\termnine\redLL{\strnat{0}}\termten\termnine$.
\item
  If $\termone=\ap{(\ca{\varone}{\termfive})}{\cm{\termsix}}$
  and $\termtwo=\sbst{\termfive}{\varone}{\termsix}$, then 
  $\termthree$ is in the form $\ap{(\ca{\varone}{\termeight})}\cm{\termnine}$ where
  $\termfive\redLLnext\termeight$ and $\termsix\redLLinf\termnine$, and
  then we can apply Lemma~\ref{lemma:commam} obtaining that
  $\termtwo\redLLnext\termfour=\sbst{\termeight}{\varone}{\termnine}$.
  On the other hand, $\termthree\redLL{\strnat{0}}\termfour$.
\end{varitemize}
\end{proof}
But there is even more: $\redLL{\strnat{0}}$ and $\redLLnext$ 
commute.
\begin{lemma}[Postponement]\label{lemma:commut}
If $\termone\in\tms{\LLLwot}$, $\termone\redLLnext\termtwo\redLL{\strnat{0}}\termthree$, 
then there is $\termfour\in\LLLwot$ such that
$\termone\redLL{\strnat{0}}\termfour\redLLnext\termthree$.
\end{lemma}
\begin{proof}
Again a coinduction on the structure of $\termone$. Some interesting
cases:
\begin{varitemize}
\item
  We can exclude the case in which $\termone=\cm{\termfive}$, because
  in that case also $\termtwo$ would be a coinductive boxes, and 
  coinductive boxes are $\redLL{\strnat{0}}$-normal forms.
\item
  If $\termone=\ap{(\ca{\varone}{\termfive})}{\cm{\termsix}}$,
  $\termtwo=\ap{(\ca{\varone}{\termseven})}{\cm{\termeight}}$
  (where $\termfive\redLLnext\termseven$ and $\termsix\redLLinf\termeight$)
  and $\termthree=\sbst{\termseven}{\varone}{\termeight}$, 
  then Lemma~\ref{lemma:commam} ensures that 
  $\termfour=\sbst{\termfive}{\varone}{\termsix}$ is such
  that $\termfour\redLLnext\termthree$.
\end{varitemize}
\end{proof}
One-step reduction is not in general confluent in infinitary $\lambda$-calculi.
However, $\redLL{\strnat{0}}$ indeed is:
\begin{proposition}\label{prop:confluenceone}
If $\termone\in\tms{\LLLwot}$, $\termone\redLL{\strnat{0}}^*\termtwo$, and
$\termone\redLL{\strnat{0}}^*\termthree$, then there is $\termfour$ such
that $\termtwo\redLL{\strnat{0}}^*\termfour$ and $\termthree\redLL{\strnat{0}}^*\termfour$.
\end{proposition}
\begin{proof}
This can be proved with standard techniques, keeping in mind that
in an inductive abstraction $\ia{\varone}{\termone}$, the variable
$\varone$ occurs finitely many times in $\termone$.
\end{proof}
The last two lemmas are of a techincal nature, but can be proved
by relatively simple arguments:
\begin{lemma}\label{lemma:reflextrans}
Both $\redLLnext$ and $\redLLinf$ are reflexive.
\end{lemma}
\begin{proof}
Easy.
\end{proof}
\begin{lemma}\label{lemma:onenext}
If $\termone\in\tms{\LLLwot}$ and $\termone\redLL{\strnat{\natone}}\termtwo$
(where $\natone\geq 1$), then $\termone\redLLnext\termtwo$.
\end{lemma}
\begin{proof}
Easy, given Lemma~\ref{lemma:reflextrans}
\end{proof}
We are finally able to prove the Confluence Theorem:
\begin{proof}[Proof of Theorem~\ref{theo:conf}]
\newcommand{\pairone}{\alpha} \newcommand{\pairtwo}{\beta} We will
show how to associate a term $\termfour=\funone(\pairone)$ to any pair
in the form
$\pairone=(\termone\redLLinf\termtwo,\termone\redLLinf\termthree)$ or
in the form
$\pairone=(\termone\redLLnext\termtwo,\termone\redLLnext\termthree)$.
The function $\funone$ is defined by coinduction on the structure of
the two proofs in $\pairone$. This will be done in such a way that in
the first case $\termtwo\redLLinf\funone(\pairone)$ and
$\termthree\redLLinf\funone(\pairone)$, while in the second case
$\termtwo\redLLnext\funone(\pairone)$ and
$\termthree\redLLnext\funone(\pairone)$. If $\pairone$ is
$(\termone\redLLinf\termtwo,\termone\redLLinf\termthree)$, then by
definition, $\termone\redLLwod^*\termfive\redLLnext\termtwo$ and
$\termone\redLLwod^*\termsix\redLLnext\termthree$.  Exploiting
Lemma~\ref{lemma:onenext}, Lemma~\ref{lemma:reflextrans}, and
Lemma~\ref{lemma:commut}, one obtains that there exist $\termseven$ and
$\termeight$ such that
$\termone\redLL{\strnat{0}}^*\termseven\redLLnext\termtwo$ and
$\termone\redLL{\strnat{0}}^*\termeight\redLLnext\termthree$.  By
Proposition~\ref{prop:confluenceone}, one obtains that there is
$\termnine$ with $\termseven\redLL{\strnat{0}}^*\termnine$ and
$\termeight\redLL{\strnat{0}}\termnine$. By repeated application of
Lemma~\ref{lemma:noninterference} and Lemma~\ref{lemma:reflextrans},
one can conclude there are $\termten$ and $\termeleven$ such that
$\termtwo\redLL{\strnat{0}}^*\termten$, $\termnine\redLLnext\termten$,
$\termnine\redLLnext\termeleven$ and
$\termthree\redLL{\strnat{0}}^*\termeleven$.  Now, let $\funone(\pairone)$ be
just
$\funone(\termnine\redLLnext\termten,\termnine\redLLnext\termeleven)$.
If, on the other hand, $\pairone$ is
$(\termone\redLLnext\termtwo,\termone\redLLnext\termthree)$, we can
define $\funone$ by induction on the proof of the two statements
where, however, we are only interested in the last thunk of
inductive rule instances. This is done in a natural way. As an
example, if $\termone$ is an application $\ap{\termfour}{\termfive}$,
then clearly $\termtwo$ is $\ap{\termsix}{\termseven}$
and $\termthree$ is $\ap{\termeight}{\termnine}$, where
$\termfour\redLLnext\termsix$, $\termfour\redLLnext\termeight$
$\termfive\redLLnext\termseven$, and 
$\termfive\redLLnext\termnine$; moreover, $\funone(\pairone)$
is the term 
$\ap{\funone(\termfour\redLLnext\termsix,\termfour\redLLnext\termeight)}
   {\funone(\termfive\redLLnext\termseven,\termfive\redLLnext\termnine)}$.
Notice how the function $\funone$ is well defined, being a guarded
recursive function on sets defined as greatest fixed points. 
\end{proof}
Confluence and Weak Normalisation together imply that normal forms are
unique:
\begin{corollary}[Uniqueness of Normal Forms]
  Every term $\termone\in\tms{\LLLwot}$ has a unique normal form.
\end{corollary}
Strangely enough, even if every term $\termone$ has a normal form
$\termtwo$, it is not guaranteed to reduce to it in every reduction
order, simply because one can certainly choose to ``skip'' certain
depths while normalising. In this sense, level-by-level sequences are
normalising: they are not allowed to go to depth $\natone>\nattwo$ if
there is a redex at depth $\nattwo$.
\subsection{Expressive Power}\label{sect:ll4sexppow}
At this point, one may wonder whether $\LLLwot$ is well-behaved simply
because its expressive power is simply too low. Although at present we
are not able to characterise the class of functions which can be
represented in it, we can already give some interesting observations
on its expressive power.

First of all, let us observe that the inductive fragment of $\LLLwot$
(i.e. the subsystem obtained by dropping coinductive boxes) is
complete for polynomial time computable functions on finite strings,
although inputs need to be represented as Church numerals for the
result to hold: this is a consequence of polytime completeness for
\SLL~\cite{Lafont04TCS}.

About the ``coinductive'' expressive power of $\LLLwot$, we note that
a form of guarded recursion can indeed be expressed, thanks to the
very liberal exponential discipline of $\FLL$. Consider the term
$\termone=\ca{\varone}{\la{\vartwo}{\ca{\varthree}{\vartwo\cm{(\varone\varone\varthree\cm{\varthree})}}}}$
and define $X$ to be $\ap{\termone}{\cm{\termone}}$. One can easily
verify that for any (closed) term $\termtwo$, the term
$X\termtwo\cm{\termtwo}$ reduces in three steps to
$\termtwo\cm{(X\termtwo\cm{\termtwo})}$.  In other words, then, $X$ is
indeed a fixed point combinator which however requires the argument
functional to be applied to it twice. 

The two observations above, taken together, mean that $\LLLwot$ is,
at least, capable of expressing all functions from $(\BB^*)^\infty$
to $(\BB^*)^\infty$ such that for each $\natone$, the
string at position $\natone$ in the output stream can be computed
in polynomial time from the string at position $\natone$ in the
input stream. Whether one could go (substantially) beyond this is an interesting
problem that we leave for further work. One cannot, however,
go too far, since the \FLL\ exponential discipline imposes that
all typable stream functions are causal, i.e., for each $\natone$,
the value of the output at position $\natone$ only depends
on the input positions \emph{up to} $\natone$, at least if one
encodes streams as in Section~\ref{sect:llwotspl}.

\section{Further Developments}
We see this work only as a first step towards understanding how linear
logic can be useful in taming the complexity of infinitary rewriting
in the context of the $\lambda$-calculus. There are at least three
different promising research directions that the results in this paper
implicitly suggest. All of them are left for future work, and are
outside the scope of this paper.
\paragraph*{Semantics} 
It would be interesting to generalise those semantic frameworks which
work well for ordinary linear logic and $\lambda$-calculi to $\LLwot$. One example
is the so-called relational model of linear logic, in which formulas
are interpreted as sets and morphisms are interpreted as binary
relations. Noticeably, the exponential modality is interpreted by
forming power multisets. Since the only kind of infinite regression we
have in $\LLwot$ is the one induced by coinductive boxes, it seems
that the relation model should be adaptable to the calculus described
here. Similarly, game semantics~\cite{AJM00IC} and the geometry of
interaction~\cite{Girard88} seem to be well-suited to model infinitary
rewriting. 
\paragraph*{Types} 
The calculus $\LLwot$ is untyped. Incepting types into it would first
of all be a way to ensure the absence of deadlocks (consider, as an
example, the term
$\ap{(\ia{\varone}{\varone})}{(\cm{\termone})}$). The natural
candidate for a framework in which to develop a theory of types for
$\LLwot$ is the one of recursive types, given their inherent relation
with infinite computations. 
Another challenge could be adapting linear
dependent types~\cite{DalLago11LICS} to an infinitary setting.
\paragraph*{Implicit Complexity} 
One of the most interesting applications of the linearisation of
$\Lwot$ as described here could come from implicit complexity, whose
aim is characterising complexity classes by logical systems and
programming languages without any reference to machine models nor to
combinatorial concepts (e.g. polynomials). We think, in particular,
that subsystems of $\LLwot$ would be ideal candidates for
characterising, e.g. type-2 polynomial time operators. This, however,
would require a finer exponential discipline, e.g. an inductive-coinductive
generalisation of the bounded exponential modality~\cite{GSS92}.
\section{Related Work}
Although this is arguably the first paper explicitly combining ideas
coming from infinitary rewriting with resource-consciousness in the
sense of linear logic, some works which are closely related to ours,
but having different goals, have recently appeared.

First of all, one should mention Terui's work on computational
ludics~\cite{Terui11TCS}: there, designs (i.e. the ludics' counterpart
to proofs) are meant to both capture syntax (proofs) and semantics
(functions), and are thus infinitary in nature. However, the overall
objective in~\cite{Terui11TCS} is different from ours: while we want
to stay as close as possible to the $\lambda$-calculus so as to
inspire the design of techniques guaranteeing termination of programs
dealing with infinite data structures, Terui's aim is to better
understand usual, finitary, computational complexity. Technically, the
main difference is that we focus on the exponentials and let them be
the core of our approach, while computational ludics is strongly based
on focalisation: time passes whenever polarity changes.

Another closely related work is a recent one by
Mazza~\cite{Mazza12LICS}, that shows how the ordinary, finitary,
$\lambda$-calculus can be seen as the metric completion of a much
weaker system, namely the affine $\lambda$-calculus. Again, the main
commonalities with this paper are on the one hand the presence of
infinite terms, and on the other a common technical background, namely
that of linear logic. Again, the emphasis is different: we,
following~\cite{Kennaway97TCS}, somehow aim at going \emph{beyond}
finitary $\lambda$-calculus, while Mazza's focus is on the
subrecursive, finite world: he is not even concerned with reduction of
infinite length.

If one forgets about infinitary rewriting, linear logic has already
been shown to be a formidable tool to support the process of isolating
classes of $\lambda$-terms having good, quantitative normalisation
properties. One can, for example, cite the work by Baillot and
Terui~\cite{Baillot09IC} or the one by Gaboardi and Ronchi
here~\cite{Gaboardi07CSL}. This paper can be seen as a natural step
towards transferring these techniques to the realm of infinitary
rewriting.

Finally, among the many works on type-theoretical approaches to
termination and productivity, the closest to ours is certainly the
recent contribution by Cave et al.~\cite{Cave14}: our treatment of the
coinductive modality is very reminiscent to their way of handling LTL
operators.
\section*{Acknowledgment}
The author would like to thank Patrick Baillot, Marco Gaboardi and
Olivier Laurent for useful discussions about the topics of this
paper. The author is partially supported by the ANR project 12IS02001
PACE and the ANR project 14CE250005 ELICA.
\bibliographystyle{abbrv}
\bibliography{biblio} 
\end{document}

%% file: macros.tex
\usepackage{amsmath}

\usepackage{a4wide,url}
\usepackage{pdfsync} 
\usepackage{microtype}
\usepackage{mathrsfs}
\usepackage{amsmath}
\usepackage{amssymb}
\usepackage{amsthm}
\usepackage{stmaryrd}
\usepackage{xspace}
\usepackage{color}
\usepackage{proof}

\newcommand{\midd}{\; \; \mbox{\Large{$\mid$}}\;\;}

\newcommand{\BB}{\mathbb{B}}
\newcommand{\NN}{\mathbb{N}}

\renewcommand{\SS}{\mathbb{S}}

\newcommand{\natone}{n}
\newcommand{\nattwo}{m}
\newcommand{\natthree}{p}
\newcommand{\natfour}{r}

\newcommand{\bitone}{a}
\newcommand{\bittwo}{b}
\newcommand{\bitthree}{c}

\newcommand{\funone}{f}

\newcommand{\setone}{X}

\newcommand{\funsetone}{\mathcal{X}}

\newcommand{\charone}{a}

\newcommand{\stsone}{Q}
\newcommand{\stone}{q}

\newcommand{\mdtone}{\mathcal{M}}

\newcommand{\sysone}{\mathsf{S}}
\newcommand{\expone}{A}
\newcommand{\exptwo}{B}
\newcommand{\jsone}{\mathcal{S}}
\newcommand{\funcone}{F}
\newcommand{\functwo}{G}
\newcommand{\functhree}{H}
\newcommand{\indrul}[1]{I_{#1}}
\newcommand{\coindrul}[1]{C_{#1}}
\newcommand{\powset}[1]{\mathcal{P}(#1)}
\newcommand{\cd}[1]{\mathbb{C}_{#1}}

\newcommand{\LLwotfin}{\ell\Lambda}
\newcommand{\LLwot}{\ell\Lambda_{\infty}}
\newcommand{\LLLwot}{\ell\Lambda_{\infty}^{\mathsf{4S}}}
\newcommand{\Lwot}{\Lambda_{\infty}}
\newcommand{\Lwotp}[3]{\Lambda_{#1#2#3}}

\newcommand{\Haskell}{\textsf{Haskell}}
\newcommand{\ML}{\textsf{ML}}

\newcommand{\SLL}{\textsf{SLL}}
\newcommand{\FLL}{\textsf{4LL}}

\newcommand{\functione}{\mathcal{F}}

\newcommand{\judgone}{J}
\newcommand{\judgtwo}{H}

\newcommand{\termone}{M}
\newcommand{\termtwo}{N}
\newcommand{\termthree}{L}
\newcommand{\termfour}{P}
\newcommand{\termfive}{Q}
\newcommand{\termsix}{R}
\newcommand{\termseven}{S}
\newcommand{\termeight}{X}
\newcommand{\termnine}{Y}
\newcommand{\termten}{Z}
\newcommand{\termeleven}{W}
\newcommand{\varone}{x}
\newcommand{\vartwo}{y}
\newcommand{\varthree}{z}

\newcommand{\varsetone}{\mathcal{V}}
\newcommand{\FV}[1]{\mathit{FV}(#1)}
\newcommand{\sbst}[3]{#1\{#2/#3\}}

\newcommand{\conone}{\Gamma}
\newcommand{\contwo}{\Delta}
\newcommand{\conthree}{\Sigma}
\newcommand{\lconone}{\Theta}
\newcommand{\lcontwo}{\Xi}
\newcommand{\lconthree}{\Psi}
\newcommand{\lconfour}{\Phi}
\newcommand{\liconone}{\Upsilon}
\newcommand{\licontwo}{\Pi}
\newcommand{\emcon}{\emptyset}
\newcommand{\sjudg}[1]{\vdash #1}
\newcommand{\tjudg}[2]{#1\vdash #2}
\newcommand{\ptjudg}[3]{#1\vdash_{#2} #3}
\newcommand{\dptjudg}[4]{#1\vdash_{#2}^{#3} #4}
\newcommand{\envleq}{\prec}

\newcommand{\LLvarl}{(\mathsf{vl})}
\newcommand{\LLvari}{(\mathsf{vi})}
\newcommand{\LLvarc}{(\mathsf{vc})}
\newcommand{\LLvard}{(\mathsf{vd})}
\newcommand{\LLvara}{(\mathsf{va})}
\newcommand{\LLapp}{(\mathsf{a})}
\newcommand{\LLlaml}{(\mathsf{ll})}
\newcommand{\LLlami}{(\mathsf{li})}
\newcommand{\LLlamc}{(\mathsf{lc})}
\newcommand{\LLlmi}{(\mathsf{mi})}
\newcommand{\LLlmc}{(\mathsf{mc})}

\newcommand{\Lvar}{(\mathsf{v})}
\newcommand{\Lapp}{(\mathsf{a})}
\newcommand{\Llam}{(\mathsf{l})}
\newcommand{\Lfuni}{(\mathsf{fi})}
\newcommand{\Lfunc}{(\mathsf{fc})}
\newcommand{\Largi}{(\mathsf{ai})}
\newcommand{\Largc}{(\mathsf{ac})}
\newcommand{\Llami}{(\mathsf{li})}
\newcommand{\Llamc}{(\mathsf{lc})}

\newcommand{\patone}{p}

\newcommand{\funsym}{\mathbf{F}}
\newcommand{\argsym}{\mathbf{A}}
\newcommand{\lamsym}{\mathbf{L}}

\newcommand{\isym}{\downarrow}
\newcommand{\csym}{\uparrow}
\newcommand{\dsym}{\#}
\newcommand{\bsym}{!}
\newcommand{\asym}{\updownarrow}
\newcommand{\psym}[1]{\updownarrow_{#1}}

\newcommand{\ap}[2]{#1#2}
\newcommand{\ab}[2]{\lambda #1.#2}
\newcommand{\la}[2]{\lambda #1.#2}
\newcommand{\ia}[2]{\lambda\!\isym\!#1.#2}
\newcommand{\ca}[2]{\lambda\!\csym\!#1.#2}
\newcommand{\pa}[3]{\lambda\!\psym{#1}\!#2.#3}
\newcommand{\na}[2]{\lambda\bsym #1.#2}
\newcommand{\im}[1]{\isym\!#1}
\newcommand{\cm}[1]{\csym\!#1}
\newcommand{\dm}[1]{\dsym\!#1}
\newcommand{\nm}[1]{\bsym #1}
\newcommand{\am}[1]{\asym\!#1}
\newcommand{\pam}[2]{\psym{#1}\!#2}

\newcommand{\derone}{\pi}
\newcommand{\dertwo}{\rho}

\newcommand{\redLLinf}{\Rightarrow}
\newcommand{\redLLinflbl}{\Rightarrow_{\mathit{lbl}}}
\newcommand{\redLLnext}{\leadsto}
\newcommand{\redLLnextlbl}{\leadsto_{\mathit{lbl}}}
\newcommand{\redLL}[1]{\rightarrow_{#1}}
\newcommand{\redLLbas}{\mapsto}
\newcommand{\redLLwod}{\rightarrow}
\newcommand{\redLLwodlbl}{\rightarrow_{\mathit{lbl}}}

\newcommand{\level}[1]{\mathit{level}(#1)}

\newcommand{\redlambda}[1]{\rightarrow_{#1}}

\newcommand{\ptms}{\mathbb{T}}
\newcommand{\tms}[1]{\mathbb{T}_{#1}}

\newcommand{\ctxone}{C}

\newcommand{\emctx}{[\cdot]}

\newcommand{\lctxone}{G}

\newcommand{\pctxone}[1]{C_{#1}}
\newcommand{\pctxtwo}[1]{D_{#1}}

\newcommand{\syneq}{\equiv}

\newcommand{\actx}[2]{#1[#2]}

\newcommand{\peremb}[2]{\langle #1\rangle_{#2}}
\newcommand{\imperemb}[3]{\langle #1\rangle_{#2#3}}

\newcommand{\psize}[2]{||#1||_{#2}}

\newcommand{\fpc}{Y}

\newcommand{\wei}[3]{W^{#1}_{#2}(#3)}
\newcommand{\twei}[2]{W_{#1}(#2)}
\newcommand{\df}[2]{D_{#1}(#2)}

\newcommand{\nfo}[2]{\mathit{NFO}(#1,#2)}

\newcommand{\alpone}{\Sigma}

\newcommand{\fins}[1]{#1^*}
\newcommand{\infs}[1]{#1^\omega}
\newcommand{\fininfs}[1]{#1^\infty}
\newcommand{\pfininfs}[2]{#1^{#2}}

\newcommand{\starone}{\mathfrak{a}}
\newcommand{\startwo}{\mathfrak{b}}

\newcommand{\strone}{s}
\newcommand{\strtwo}{t}

\newcommand{\fsemb}[1]{\langle #1\rangle_*}
\newcommand{\isemb}[1]{\langle #1\rangle_\omega}
\newcommand{\psemb}[2]{\langle #1\rangle_{#2}}

\newcommand{\sigone}{\Phi}

\newcommand{\funsymone}{\mathtt{f}}
\newcommand{\funsymtwo}{\mathtt{g}}

\newcommand{\emstr}{\varepsilon}

\newcommand{\ftermone}{t}
\newcommand{\ftermtwo}{s}

\newcommand{\fa}[1]{\mathbb{FA}(#1)}
\newcommand{\fc}[1]{\mathbb{FC}(#1)}

\newcommand{\faemb}[2]{\langle #2\rangle_{#1}^\mathbb{A}}
\newcommand{\fcemb}[2]{\langle #2\rangle_{#1}^\mathbb{C}}

\newenvironment{varitemize}
{
\begin{list}{\labelitemi}
{\setlength{\itemsep}{0pt}
 \setlength{\topsep}{0pt}
 \setlength{\parsep}{0pt}
 \setlength{\partopsep}{0pt}
 \setlength{\leftmargin}{15pt}
 \setlength{\rightmargin}{0pt}
 \setlength{\itemindent}{0pt}
 \setlength{\labelsep}{5pt}
 \setlength{\labelwidth}{10pt}
}}
{
 \end{list} 
}

\newcounter{numberone}

\newenvironment{varenumerate}
{
\begin{list}{\arabic{numberone}.}
{
  \usecounter{numberone}
  \setlength{\itemsep}{0pt}
  \setlength{\topsep}{0pt}
  \setlength{\parsep}{0pt}
  \setlength{\partopsep}{0pt}
  \setlength{\leftmargin}{15pt}
  \setlength{\rightmargin}{0pt}
  \setlength{\itemindent}{0pt}
  \setlength{\labelsep}{5pt}
  \setlength{\labelwidth}{15pt}
}}
{
\end{list} 
}

\newcounter{numbertwo}

\newtheorem{theorem}{Theorem}
\newtheorem{lemma}{Lemma}
\newtheorem{proposition}{Proposition}

\newtheorem{corollary}{Corollary}

